\documentclass[notitlepage,11pt]{article}

\usepackage{amsmath, amsthm, amsfonts} 
\usepackage{amssymb}
\usepackage[margin=1in]{geometry}
\usepackage{authblk}

\usepackage[numbers]{natbib}
\usepackage{savesym}
\savesymbol{NO}
\usepackage{graphicx}
\usepackage{hyperref}
\usepackage{booktabs} 
\usepackage{multirow}
\usepackage{epstopdf}
\usepackage{pdfpages}
\usepackage{cancel}
\usepackage{hyperref}
\usepackage[noend]{algpseudocode}
\usepackage{algorithm}
\usepackage{algorithmicx}
\usepackage{array}
\usepackage{subfigure}
\usepackage{mathrsfs}
\usepackage{upgreek}
\usepackage{eufrak}
\usepackage{lscape}
\usepackage{tkz-graph}
\usepackage{moresize}
\usepackage{rotating}
\usepackage{afterpage}

\algdef{SxnE}[WHILE]{While}{EndWhile}[1]{\algorithmicwhile\ #1\ \algorithmicdo}
\algdef{SxnE}[PROCEDURE]{Procedure}{EndProcedure}[2]{\algorithmicprocedure\ \textproc{#1}\ifthenelse{\equal{#2}{}}{}{ (#2)}}
\algdef{SxnE}[FOR]{For}{EndFor}[1]{\algorithmicfor\ #1\ \algorithmicdo}
\algdef{SxnE}[IF]{If}{EndIf}[1]{\algorithmicif\ #1\ \algorithmicthen}
\algdef{cxnE}{IF}{Else}{EndIf}

\usepackage{tablefootnote}
\usepackage{footnote}
\usetikzlibrary{arrows,positioning,automata}
\makesavenoteenv{tabular}

\newtheorem{theorem}{Theorem}
\newtheorem{proposition}[theorem]{Proposition}

\newtheorem{property}{Property}

\theoremstyle{remark}

\newtheorem{example}{Example}

\DeclareMathOperator{\proj}{Proj}

\newcommand{\su}[2][1]{\ensuremath{\displaystyle{\sum_{#1}^{#2}}}}

\newcommand{\R}{\mathbb{R}}
\newcommand{\Z}{\mathbb{Z}}

\newcommand{\lcup}{\cup} 
\newcommand{\xbar}{\overline{x}}
\newcommand{\ybar}{\overline{y}}
\newcommand{\xp}{x^{\prime}}
\newcommand{\yp}{y^{\prime}}

\renewcommand{\P}{\mathbb{P}}
\renewcommand{\xi}{f}
\newcommand{\Ps}{\Om}
\newcommand{\ps}{\om}
\newcommand{\Om}{\Omega}
\newcommand{\om}{\omega}

\newcommand{\subtree}{\textit{subtree}}

\newcommand*\samethanks[1][\value{footnote}]{\footnotemark[#1]}

\title{Efficient storage of Pareto points in biobjective mixed integer programming}

\author[1]{Nathan Adelgren\thanks{Part of the work was carried out at Department of Mathematical Sciences, Clemson University}\textsuperscript{,}\thanks{Partial support by ONR grant  N00014-16-1-2725 is acknowledged}}
\affil[1]{Department of Mathematics and Computer Science, Edinboro University, Edinboro, PA 16444, \texttt{nadelgren@edinboro.edu}}
\author[2]{Pietro Belotti}
\affil[2]{FICO, International Square, Starley Way, Birmingham B37 7GN, United Kingdom, 
\texttt{pietrobelotti@fico.com}}
\author[3]{Akshay Gupte\samethanks}
\affil[3]{Department of Mathematical Sciences, Clemson University, Clemson, SC 29634,
\texttt{agupte@clemson.edu}}

\date{}

\begin{document}

\maketitle

\begin{abstract}
In biobjective mixed integer linear programs (BOMILPs), two linear objectives are minimized over a polyhedron while restricting some of the variables to be integer. Since many of the techniques for finding or approximating the Pareto set of a BOMILP use and update a subset of nondominated solutions, it is highly desirable to efficiently store this subset. We present a new data structure, a variant of a binary tree that takes as input points and line segments in $\R^2$ and stores the nondominated subset of this input. When used within an exact solution procedure, such as branch-and-bound (BB), at termination this structure contains the set of Pareto optimal solutions. 

We compare the efficiency of our structure in storing solutions to that of a dynamic list which updates via pairwise comparison. Then we use our data structure in two biobjective BB techniques available in the literature and solve three classes of instances of BOMILP, one of which is generated by us. The first experiment shows that our data structure handles up to $10^7$ points or segments much more efficiently than a dynamic list. The second experiment shows that our data structure handles points and segments much more efficiently than a list when used in a BB.
\end{abstract}



\section{Introduction}

Biobjective mixed integer linear programs (BOMILP) have the following form:
\begin{equation}\label{BOMILP}
\begin{array}{cl}
\min_{x,y} & f(x,y):= \left[f_1(x,y) := c_1^{\top}x+d_1^{\top}y, \: f_2(x,y) := c_2^{\top}x+d_2^{\top}y\right]\\
\text{s.t.} & (x,y) \in P_{I}:= \{(x,y)\in\R^{m}\times\Z^{n}\colon Ax+By\leq b\},
\end{array}
\end{equation}
where $P_I$ is a bounded set.  
Define as $\Ps := \{\ps\in \R^2: \ps=f(x,y) \,\,\forall (x,y)\in P_I\}$ the collection of all points in $\R^2$ which can be obtained using the objective function values of feasible solutions to \eqref{BOMILP}. We refer to the space $\R^2$ containing $\Ps$ as the \emph{objective space}.

For any two vectors $v^1,v^2\in \R^2$ we use the following notation: $v^1 \leqq v^2$ if $v_i^1 \leq v_i^2$ for $i = 1, 2$; $v^1 \leq v^2$ if $v^1 \leqq v^2$ and
$v^1 \neq v^2$; and $v^1 < v^2$ if $v_i^1 < v_i^2$ for $i = 1, 2$. Given distinct $(\xbar,\ybar),(\xp,\yp)\in P_{I}$, we say that $f(\xbar,\ybar)$ \emph{dominates} $f(\xp,\yp)$ if $f(\xbar,\ybar) \le f(\xp,\yp)$. 
This dominance is \emph{strong} if $f(\xbar,\ybar) < f(\xp,\yp)$; otherwise it is \emph{weak}. 
A point $(\overline{x},\overline{y}) \in P_I$ is (\emph{weakly}) \emph{efficient} if $\nexists \,\, (x^{\prime},y^{\prime}) \in P_I$ such that $f(x^{\prime},y^{\prime})$ (strongly) dominates $f(\overline{x},\overline{y})$. The set of all efficient solutions in $P_I$ is denoted by $X_E$. A point $\overline{\ps} = f(\overline{x},\overline{y})$ is called \emph{Pareto optimal} if $(\overline{x},\overline{y}) \in X_{E}$. Given ${\Ps}' \subseteq \Ps$ we say that $\ps' \in \Ps'$ is \emph{nondominated} in $\Ps'$ if $\nexists \,\, \ps''\in \Ps'$ such that $\ps''$ dominates $\ps'$. Note that Pareto optimal points are nondominated in $P_I$. We consider a BOMILP solved when 
the set of Pareto optimal points $\Om_P := \{\om\in\R^2 : \om = f(x,y) \,\, \forall (x,y)\in X_{E} \}$ is found. 

It is known \citep{ehrgott2005multicriteria} that a biobjective LP (BOLP) can be solved by taking convex combinations of $f_{1}(\cdot)$ and $f_{2}(\cdot)$ and solving a finite number of LPs. Thus for BOLP, the set of Pareto points can be characterized as $\Omega_P = \{(\xi_{1},\xi_{2})\in\R^2\colon \xi_{2}=\psi(\xi_{1})\}$ where $\psi(\cdot)$ is a continuous, convex, piecewise linear function obtained using extreme points of the feasible region. For biobjective integer programs (BOIPs) it is known that $\Omega_P$ is a finite set of points in $\R^{2}$. Now consider the case of BOMILP. Let $Y = \proj_{y} P_{I}$ be the set of integer feasible subvectors to \eqref{BOMILP}. Since $P_{I}$ is bounded, we have $Y = \{y^{1},\ldots,y^{k}\}$ for some finite $k$. Then for each $y^{i}\in Y$ there is an associated BOLP, referred to as a \emph{slice problem} and denoted $\P(y^i)$, obtained by fixing $y=y^{i}$ in \eqref{BOMILP}:
\begin{equation}\label{slice}
\begin{array}{rl}
\P(y^i) \qquad {\min_x} & \{ f_1(x) = c_1^{\top}x+d_1^{\top}y^i, \: f_2(x) = c_2^{\top}x+d_2^{\top}y^i\}\\
\text{s.t.} & Ax \leq b - By^i.\\
\end{array}
\end{equation}
Problem $\P(y^i)$ has a set of Pareto solutions $\Omega_i := \{(\xi_{1},\xi_{2})\in\R^2\colon \xi_{2} = \psi_{i}(\xi_{1}) \}$, where $\psi_{i}(\cdot)$ is a continuous convex piecewise linear function as explained before. Then $\Omega_P \subseteq {\mathop{\lcup}_{i=1}^k \Omega_i}$ and this inclusion is strict in general. In particular, we have: 
\begin{equation}\Omega_P = \mathop{\lcup}_{i=1}^k\left( \Omega_i \setminus \mathop{\lcup}_{j\neq i}\left(\Omega_j+\R^2_+\setminus\{\textbf{0}\}\right)\right).
\end{equation} 
Such union of sets is not, in
general, represented by a convex piecewise linear function. Figure \ref{fig:store_data} shows an example with $k=5$.


It should be noted that finding $\Omega_P$ is not a trivial task in general. In the worst case, $\Omega_P = \cup_{i=1}^{k} \Omega_{i}$ and one may have to solve every slice problem to termination, which can have exponential complexity. For multiobjective IPs (i.e. $m=0$), \citet{de2009pareto} prove that $\Omega_P$ can be enumerated in polynomial-time for fixed $n$, which extends the well known result that single-objective IP's can be solved in polynomial-time for fixed $n$ \citep{lenstra1983integer}. We are unaware of any similar results for BOMILP.

Exact procedures for solving BOMILP with general integers have been recently the subject of intense research. Exact methods have been presented by  \citet{belotti2012biobjective,belotti2016fathoming}, \citet{boland2013biobjective}, and more recently by \citet{soylu2016exact}. \citet{ozpeynirci2010exact} give an exact method for finding {\em supported} solutions of BOMILP. Most other techniques in the literature tackle specific cases. \citet{vincent2013biobjective} improved upon the method of \citet{mavrotas2005multi} for mixed 0-1 problems. \citet{stidsen2014branch} propose a method for solving mixed 0-1 problems in which only one of the objectives contains continuous variables. \citet{belotti2012biobjective,mavrotas2005multi,stidsen2014branch} and \citet{vincent2013biobjective} are based on branch-and-bound (BB) procedures in which the Pareto set is determined by solving several BOLPs. Instead, \citet{boland2013biobjective}, \citet{ozpeynirci2010exact}, and \citet{soylu2016exact} determine the Pareto set by solving several MIPs, albeit in different ways: while, for example, \citet{boland2013biobjective} recursively partition the objective space to circumscribe subsets of Pareto points and segments, \citet{soylu2016exact} find the Pareto frontier incrementally, starting from a solution to a single-objective problem. The pure integer case has been studied for binary variables \citep{kiziltan1983algorithm}, general integers \citep{ralphs2006improved} and specific classes of combinatorial problems \citep{sourd2008multiobjective,przybylski2010149,jozefowiez2012generic}.

We present a data structure for efficiently storing a nondominated set of feasible solutions to a BOMILP. For lack of better names, we call it {\em biobjective tree}, or BoT. The BoT can be used in exact and heuristic solution procedures that aim at finding or approximating the Pareto set. Data structures such as {\em quad-trees} have been used for storing Pareto points in the past \citep{sun1996,sun2006}, although only in the pure integer case. \citet{sun1996} stored nondominated solutions using both quad-trees and dynamic lists which were updated via pairwise comparison. They showed that in the pure integer, biobjective case, dynamic lists store nondominated solutions more efficiently than quad-trees.

In Section \ref{overview} we describe the BoT, and in Section \ref{insertion} we describe the insertion function, prove its correctness and provide an example of use. In Section \ref{results} we present the results of two experiments. The first experiment shows that a BoT is able to store nondominated points and segments more efficiently than a dynamic list and can insert up to $10^7$ solutions in reasonable time. In the second experiment we utilize a BoT in the BB procedures of \citet{belotti2012biobjective} and \citet{adelgren2016} to solve specific instances of BOMILP. The results show that BoT handles points and segments much more efficiently than a list when used in a BB.

\section{Biobjective tree (BoT)}
\label{overview}

Figure \ref{gen_pts_sgmts} shows an example of solutions that might be generated when solving an instance of BOMILP. We would like to store the nondominated portion of these points and segments, as shown in Figure \ref{store_pts_sgmts}. 
\begin{figure}
\centering
\subfigure[Generated points and segments]{
\includegraphics[width=5cm]{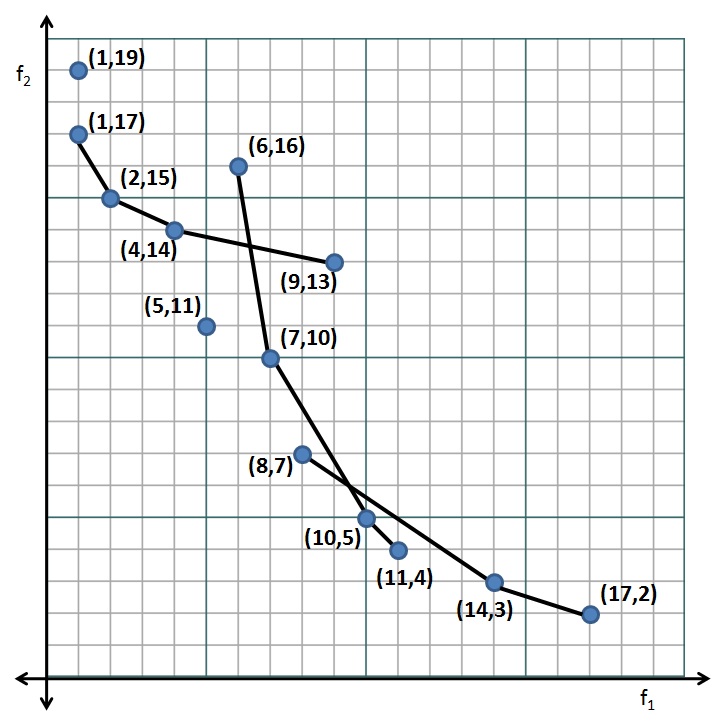}
\label{gen_pts_sgmts}
}\hspace{1cm}
\subfigure[Nondominated subset]{
\includegraphics[width=5cm]{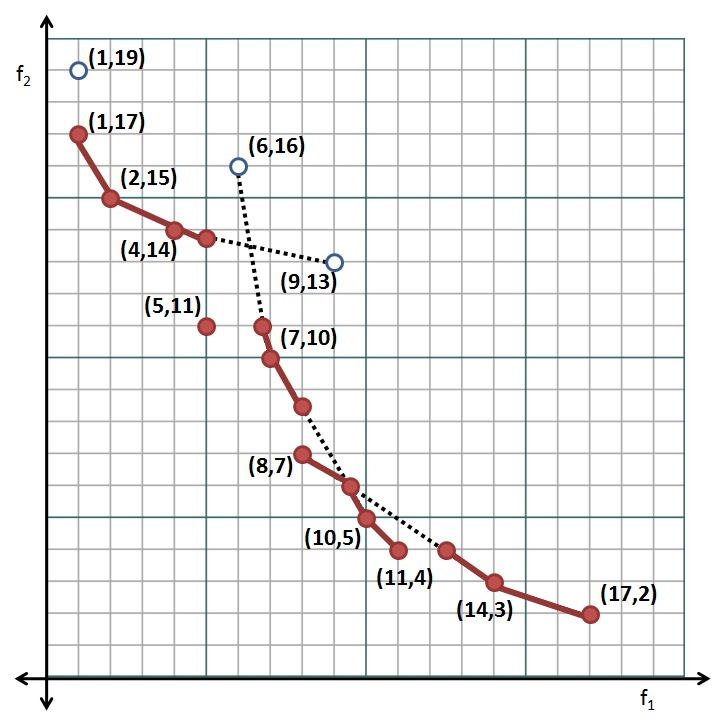} 
\label{store_pts_sgmts}
}
\caption{Example of solutions generated when solving an instance of BOMILP.}
\label{fig:store_data}
\end{figure} 
A BoT stores only the nondominated subset of the solutions regardless of the order in which they are inserted. Consider the set of solutions  in Figures~\ref{gen_pts_sgmts} and suppose that the segments connecting (1,17), (2,15), (4,14), and (9,13) are currently stored. When inserting the point (5,11), which dominates a portion of the segment connecting (4,14) and (9,13), the portion must be removed before the point is added. Similarly, when the segment connecting (6,16) and (7,10) is inserted, because a portion of this segment is dominated by (5,11), only the nondominated portion of the segment should be added.
 
A BoT is a binary search tree (BST) in which each node represents either a singleton or a line segment associated with a Pareto point or set of Pareto points in the objective space. Denote as  $\Pi$ the set of nodes in the tree. Each $\pi \in \Pi$ is defined as a triplet $\pi = (S,l,r)$, its components defined as follows:

\begin{itemize}
\item $S=(x_1,x_2,y_1,y_2)$ represents a segment $[(x_1,y_1),(x_2,y_2)]$ defined in the objective space; the segment collapses to a point if $x_1=x_2$ and $y_1=y_2$.
\item $l$ and $r$ are the left and the right child node of $\pi$, respectively.
\end{itemize}

We identify each element of a triplet $\pi$ as $\pi.S$, $\pi.l$, and $\pi.r$. Let us define operators that are used in the remainder: $\textit{size}(\pi)$ is the size of the subtree rooted at $\pi$, defined as $\textit{size}(\emptyset) = 0$ and $\textit{size}(\pi) = 1 + \textit{size}(\pi.l) + \textit{size}(\pi.r)$; also, $\textit{parent}(\pi)$ is a node $\pi'$ such that $\pi'.l=\pi$ or $\pi'.r=\pi$, if any exists, or $\emptyset$ otherwise.   Hence $\pi$ is a {\em leaf node} if $\pi.l =\pi.r = \emptyset$, or the {\em root node} if $\textit{parent}(\pi)=\emptyset$.  Finally, $\subtree(\pi)$ is the subtree rooted at $\pi$. In the remainder, even though a node $\pi$ is defined by a triplet, we sometimes use $\pi$ to refer to $\pi.S$ for ease of notation (especially with set operations), but only where we believe this does not lead to confusion. In particular, for $\pi=(S,l,r)$ the operation $\pi \cap A$, where $A\subseteq \mathbb R^2$, returns a point $\pi'=(S',l,r)$ such that $S' = S \cap A$.


A BoT contains nodes that correspond each to a segment $S = [(x_1,y_1), (x_2,y_2)]$ such that $x_1 \le x_2$ and $y_1 \ge y_2$, as otherwise $S$ can be reduced to a single point. The BoT maintains a non-strict total order $\preceq$ between nodes: for two nodes $\pi' = (S',l',r'), \pi'' = (S'',l'',r'') \in \Pi$, with $S'=[(x'_1,y'_1),(x'_2,y'_2)]$ and $S'' = [(x''_1,y''_1), (x''_2,y''_2)]$,  $\pi' \preceq \pi''$ if either $\pi'$ and $\pi''$ are the same, or $x'_2 \le  x''_1$ and $y'_2 \ge y''_1$. Our notation extends to node segments as well, i.e., $\pi' \preceq \pi''\equiv \pi'.S \preceq \pi''.S$. Also, for any $\pi\in\Pi$, we have $\pi.l \prec \pi \prec \pi.r$.

All operations on the BoT must conserve this total order; as for any BST, enumerating its sorted elements amounts to an {\em in-order} parse of the tree. Removal of a subtree from the BST and rebalancing subtrees preserves the order \citep{knuth1998art}. Most tree operations carry over to the BoT without change, but insertion, discussed in Section \ref{insertion}, is radically different: 
a BST insertion increases the tree size by one, while in a BoT inserting a single node might have a large-scale effect. For instance, an entire subtree might be deleted if the inserted segment $S$ dominates all nodes of the subtree.  Even if $S$ does not dominate any of the current nodes, the non-dominated portion of $S$ could be as many as $t+1$ disjoint segments (if $t$ nodes are stored).

This is the main point of departure with classical BSTs and other data structures (red-black trees, for instance) that have fast insertion and lookup. This also explains why other data structures such as a list or a hash-table might be less suited for this purpose: lookup and insertion in the list are rather expensive at $O(t)$ if $t$ is the number of elements in the list; insertion of a point in a hashtable and the possible deletion of a large number of elements is inefficient in a hash table for many reasons. One such reason is {\em locality} of the data: a good hash function would guarantee large separation, in the hash table, between two adjacent points or segments in the objective space. Finding all segments that are dominated by one point or segment would prove inefficient, as it would require visiting the entire data structure.



The following notation is useful in the remainder of the paper: we partition $\mathbb R^2$ relative to the segment $\pi.S = [(x_1,y_1), (x_2,y_2)]$ of a node $\pi$ into four regions $R_\alpha(\pi)$ for $\alpha \in \{\textrm{up}, \textrm{dn}, \textrm{left}, \textrm{right}\}$. We define $R_{\textrm{up}}(\pi)    = \pi.S + \{(x,y)\in \R^2: x   > 0, y   > 0\}$,  
  i.e., the set of points dominated by $\pi$; $R_{\textrm{dn}}(\pi)    = \pi.S + \{(x,y)\in \R^2: x \le 0, y \le 0\}$;
$R_{\textrm{left}}(\pi)  = \{(x,y)\in \R^2: x \le x_1, y   > y_1\}$; and finally
$R_{\textrm{right}}(\pi) = \{(x,y)\in \R^2: x   > x_2, y \le y_2\}$  (see Figure \ref{fig:partition}).

\begin{figure}
\begin{center}
\subfigure[Segment]{
\includegraphics[width=.25\textwidth]{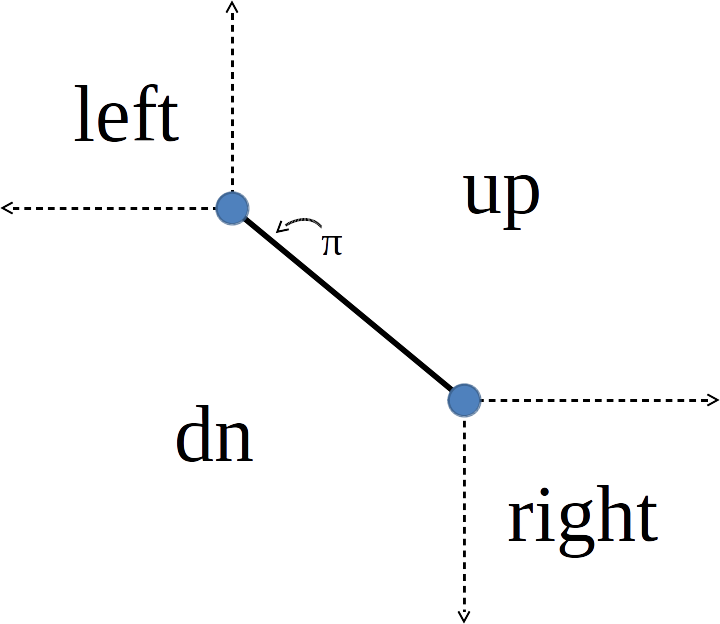}
\label{partition_sgmt}
}\hfill
\subfigure[Single point]{
\includegraphics[width=.25\textwidth]{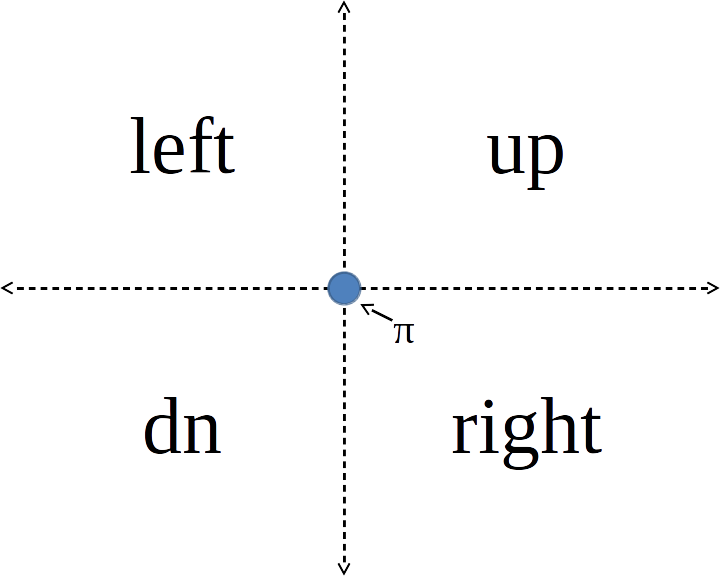}
\label{partition_pnt}
}\hfill
\subfigure[Weak domination.]{
\begin{tikzpicture}
\node[anchor=south west,inner sep=0] (image) at (0,0) {\includegraphics[width=.25\textwidth]{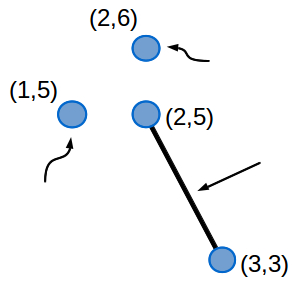}};
\node[yshift=-.85cm,xshift=-1.3cm] at (image.center) { $\pi_1$};
\node[yshift=-0.19cm,xshift=1.7cm] at (image.center) { $\pi$};
\node[yshift=1.1cm,xshift=1.1cm] at (image.center) { $\pi_2$};
\end{tikzpicture}
\label{weak_dom}
}
\caption{Partition of $\R^2$ w.r.t.\mbox{} a segment or point $\pi.S$ (a,b); example of weak domination (c).}
\label{fig:partition}
\end{center}
\end{figure}
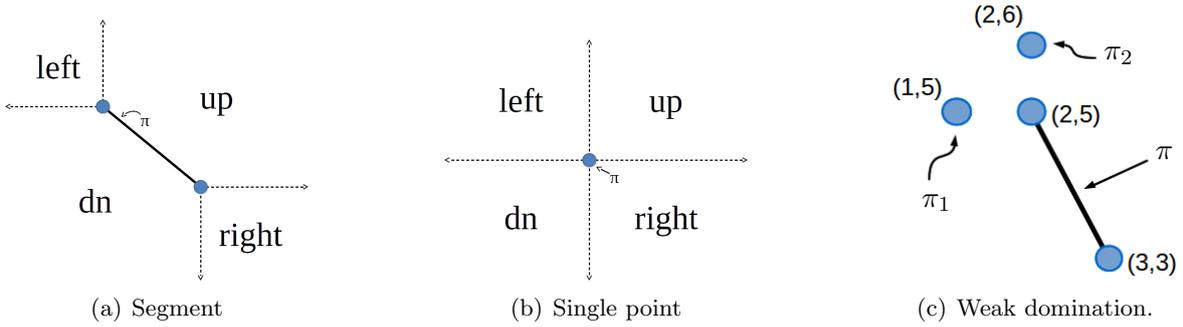


\begin{example} Suppose $\pi \in \Pi$ is defined by the segment between $(2,5)$ and $(3,3)$. Further suppose that $\pi_1$ is the point $(1,5)$, $\pi_2$ is the point $(2,6)$ and both $\pi_1$ and $\pi_2$ are inserted into the BoT. Observe Figure \ref{weak_dom}.
The point associated with node $\pi_1$ weakly dominates the left-most point of the segment associated with $\pi$ and thus $\pi_1$ should be stored. However, the point associated with node $\pi_2$ is weakly dominated by the segment associated with $\pi$ and so $\pi_2$ should not be stored. 
\end{example}

The previous example leads us to our next choice. Consider $\pi$ and $\pi_1$ in Figure \ref{weak_dom}: we should mark the segment as open at $(2,5)$, i.e., $\pi.S = ](2,5),(3,3)]$ assuming no other points dominate $(3,3)$. We do not memorize whether a segment $\pi.S$ is open or closed on either extreme. Verifying such information would lead to more complex algorithms. Instead, upon termination of the BB, for any node $\pi$ we  check the nodes immediately to the left and to the right to determine if $\pi.S$ is (partially) open.

\section{Insertion}
\label{insertion}

%

The \textsc{Insert} function is described in Algorithm \ref{insert}. It takes two inputs: a node $\pi^*$ to be inserted and a node $\pi$ which is the root of the tree. 
Because a node $\pi$ does not hold any information of all other nodes of the subtree rooted at $\pi$, the insertion of $\pi^*$ might have to be propagated both to $\pi.l$ and to $\pi.r$. For this reason, the \textsc{Insert} function is recursive. If the recursive call inserts $\pi^*$ in $\pi.l=\emptyset$ (or $\pi.r=\emptyset$), then $\pi.l$ becomes $\pi^*$. 
The main call to \textsc{Insert} is done by passing the root node of the tree, denoted as $\pi_0$, as the second argument. 

%


The function \textsc{Replace}($\pi',\pi''$) replaces $\pi'\in\Pi$ with $\pi''\in\Pi$, leaving the BoT otherwise unchanged. We instead denote with  $\pi' \leftarrow \pi''$ the process of replacing $\pi'$ and its entire subtree with $\pi''$ and its entire subtree. \textsc{RemoveNode}$(\pi)$, described in Algorithm \ref{removenode}, deletes $\pi$ if it is a leaf node, otherwise it replaces $\pi$ with a node $\pi'$ of its subtree that retains the total order. For this, $\pi'$ must be adjacent to $\pi$ to the left or right, i.e., it  is the left-most node of the right subtree of $\pi$ or right-most node of the left subtree of $\pi$.

\begin{algorithm}
\small
  \caption{Inserting a new point or segment, $\pi^*$, into a BoT at node $\pi$}
  \label{insert}
  \begin{algorithmic}[1]
    \Function{Insert}{$\pi^*, \pi$} 
      \If{$\pi^* = \emptyset$}{ \Return}\EndIf
      \If{$\pi = \emptyset$}{ \textsc{Replace}$\left(\pi,\pi^*\right)$}
      \Else\State{Define $S':= \pi.S\setminus cl\left(R_{\textrm{up}}(\pi^*)\right)$}
          \If{$S' = \emptyset$}{}
      	      \State \textsc{RemoveNode}$(\pi)$
      	      \State \textsc{Insert}$(\pi^*,\pi)$
          \Else
      	      \If{$\exists S'', S'''$ s.t.
                    $S' = S'' \cup S''' \wedge
                    cl(S'')\cap cl(S''') = \emptyset \wedge S''\preceq S'''$}
      	          \State Create new node $\pi' = (S''', \emptyset, \pi.r)$
      	          \State $\pi.r \leftarrow \pi'$
                  \State $\pi.S \leftarrow S''$
              \Else{ $\pi.S\leftarrow S'$}
              \EndIf
      	      \State \textsc{Insert}$(\pi^* \cap R_{\textrm{left}}(\pi),\pi.l)$
      	      \State \textsc{Insert}$(\pi^* \cap R_{\textrm{right}}(\pi),\pi.r)$
      	  \EndIf
      \EndIf
    \EndFunction
  \end{algorithmic}
\end{algorithm} 

We skip the trivial cases and focus on lines 5 onward in Algorithm \ref{insert}. The set $S'$ is obtained by removing from $\pi.S$ the set of points dominated by $\pi^*$. If $S'=\emptyset$, $\pi$ is entirely dominated by $\pi^*$ and can be removed; the procedure is then called recursively on the new subtree.  If $S'\neq \emptyset$, then it might be the union of at most two subsegments $S''$ and $S'''$. If so, $S'''$ becomes the segment of a new node  $\pi'$, which is assigned the right subtree of $\pi$, and $\pi$ has its segment restricted to $S''$. Otherwise $\pi.S$ is changed to $S'$. Insertion is then called recursively on the left and right subtrees of $\pi$.


The following property is equivalent to the total order mentioned in Section \ref{overview}.

\begin{property}\label{property_1}
For any $\pi\in\Pi$, all nodes in the subtree of $\pi.l$ (resp. $\pi.r$) are located completely within $R_{\textrm{left}}(\pi)$ (resp.\mbox{} $R_{\textrm{right}}(\pi)$).
\end{property}


\begin{algorithm}
\small
  \caption{Remove a node that has been shown to be dominated.}
  \label{removenode}
  \begin{algorithmic}[1]
    \Function{RemoveNode}{$\pi$} 
    \If{$\textit{size}(\pi) = 1$}{ $\pi \leftarrow \emptyset$}
    \Else
        \If{$\textit{size}(\pi.l) > \textit{size}(\pi.r)$}{ $\tilde{\pi} \leftarrow$\textsc{FindRightmostNode}($\pi.l$)}
        \Else{ $\tilde{\pi} \leftarrow$\textsc{FindLeftmostNode}($\pi.r$)}\EndIf
        \State \textsc{Replace}$\left(\pi,\tilde{\pi}\right)$
        \State \textsc{RemoveNode}($\tilde{\pi}$)
    \EndIf
    \EndFunction
  \end{algorithmic}
\end{algorithm}

For a tree with $t$ nodes at the time of an insertion, the worst-case complexity of \textsc{Insert} is $O(t)$ as the entire tree might have to be visited. For example, consider the tree containing all points $\{(i, M-i), i = 1,2,\ldots, M-1\}$ for any $M > 2$. Inserting $\pi^* = (0,0)$, which results in $\pi^*$  replacing the whole tree, requires visiting all nodes of the tree.

\subsection{Correctness}


\begin{proposition}
\label{prop1}
\textsc{Insert} removes any portion of a currently stored node $\pi$ which is dominated by an inserted node $\pi^*$. 
\end{proposition}
\begin{proof}
Obvious from steps 7 (if $S'=\emptyset$), 13 (if $S'=S''\cup S'''$), and 14 ($S'$ is a segment) of Algorithm \ref{insert}.

\end{proof}

\begin{proposition}
\label{prop2}
\textsc{Insert} adds node $\pi^*$, or a portion, to the tree if and only if it is not dominated by any node currently stored in the tree. 
\end{proposition}
\begin{proof}
(If) Let $S'$ be a portion of $\pi^*.S$ not dominated by any $\pi\in\Pi$; then $(S',\emptyset, \emptyset)$ is inserted as a leaf node after a finite number of recursive calls --- see lines 15 and 16.

(Only if) We show the contrapositive. Suppose there is a portion $\tilde S$ of $\pi^*.S$ that is dominated by $\pi\in\Pi$. 
If $\pi = \pi_0$, then \textsc{Insert}$(\pi^*, \pi)$ has the effect (line 5) of removing $\tilde S$ from the set of points to be inserted, while the two recursive calls at lines 15 and 16 are made for $\pi^*\cap R_{\textrm{left}}(\pi)$ and $\pi^*\cap R_{\textrm{right}}(\pi)$, both excluding $\tilde S$. If $\pi \neq \pi_0$ and $\pi \in \subtree(\pi_0.l)$ (resp.\mbox{} $\pi \in \subtree(\pi_0.r)$), the recursive call \textsc{Insert}$(\pi^* \cap R_{\textrm{left}}(\pi_0),\pi_0.l)$ (resp.\mbox{}  \textsc{Insert}$(\pi^* \cap R_{\textrm{right}}(\pi_0),\pi_0.r)$) guarantees that a call to \textsc{Insert}$(\pi^*, \pi)$ will be placed, i.e., $\tilde S$ will be eliminated by subsequent recursive calls.
%
\end{proof}


According to Propositions \ref{prop1} and \ref{prop2}, the end state of the BoT contains all and only nondominated solutions. 



\begin{proposition}
\label{prop_order}
\textsc{Insert} retains Property \ref{property_1}.
\end{proposition}
\begin{proof}
The result is trivial for the cases in lines 2 and 3 of Algorithm \ref{insert}. The same holds if $S'=\emptyset$ (line 6), as it is easy to prove that \textsc{RemoveNode} retains Property \ref{property_1}. If, for the node $\pi^*$ to be inserted, $\pi^*.S \subset \cup_{\pi \in \Pi} R_{\textrm{up}}(\pi)$, then $\pi^*$ is dominated and by Proposition \ref{prop2} it will not be inserted, thereby not modifying the BoT. By the same proposition, any segment $\pi^*.S$ or portion that neither dominates nor is dominated by any node will be added to the BoT so as to satisfy Property \ref{property_1}: \textsc{Insert} recursively runs lines 15 and 16 until said portion is added as a leaf.

Assume now that $\pi^*.S \cap R_{\textrm{dn}}(\pi) \neq \emptyset$ for at least one $\pi \in \Pi$.
Since $\pi^*$ is not dominated, two cases arise: (i) $\pi.S \setminus cl(R_{\textrm{up}}(\pi^*))$ is a single segment $S'$; (ii) $\pi.S \setminus cl(R_{\textrm{up}}(\pi^*))$ is the union of two disjoint segments $S''$ and $S'''$ with $S'' \preceq S'''$.

In case (i), the property holds after running line 14, as $\pi.S$ is replaced by its subset $S'$.  Case (ii) is dealt with on lines 10-13:  since $S'' \preceq S'''$, $\tilde{\pi}.S \preceq S''$ for all $\tilde{\pi} \in \subtree(\pi.l)$ and $S''' \preceq \hat{\pi}.S$ for all $\hat{\pi} \in \subtree(\pi.r)$. Thus, replacing $\pi.S$ with $S''$ ensures that Property 1 is maintained for $\pi.l$ and placing $\pi' := (S''',\emptyset,\pi.r)$ as the right child of $\pi$ ensures that the same holds for $\pi.r$. By construction $\pi\preceq\pi'$, and this concludes the proof.
%
%
\end{proof}

\subsection{Illustrative example} \label{example}
We use the points and segments in Figure \ref{gen_pts_sgmts} as input to a BoT and show a few of the nontrivial steps of developing it. Assume that these solutions are obtained from five separate slice problems and that the Pareto sets of these slice problems are: (i) the singleton (1,19), (ii) the piecewise linear curve connecting (1,17) and (9,13), (iii) the one connecting (6,16) and (11,4), (iv) the singleton (5,11), and (v) the piecewise linear curve connecting (8,7) and (17,2).
The points and segments which define these Pareto sets are inserted into the BoT in the order of (iii), (iv), (ii), (v), (i). Piecewise linear curves are inserted as individual line segments from left to right. 

The first call is on an empty tree, hence $\pi^* = [(6,16),(7,10)]$ becomes the tree. 
The next point is $\pi^* = [(7,10),(10,5)]$; note that $\pi^* \subset R_{\textrm{right}}(\pi_0)$ and should be inserted at $\pi_0.r$. Since $\pi_0.r=\emptyset$,  $\pi^*$ replaces $\pi_0.r$.
%
The insertion of  $[(10,5),(11,4)]$ is analogous.

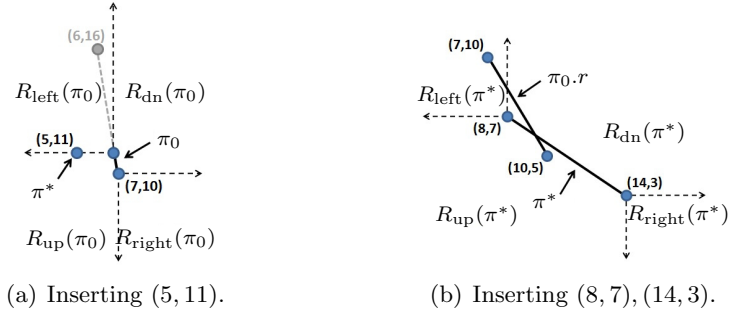
\begin{figure}
\centering
\subfigure[Inserting $(5,11)$.]{
\begin{tikzpicture}
\node[anchor=south west,inner sep=0] (image) at (0,0) {\includegraphics[width=2.5cm]{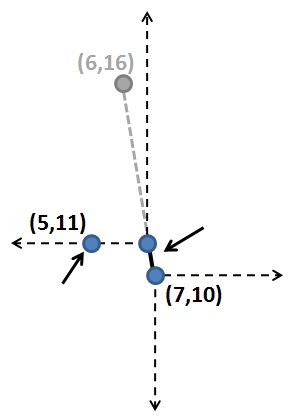}};
\node[yshift=-.15cm,xshift=0.75cm] at (image.center) {\scriptsize $\pi_0$};
\node[yshift=-.8cm,xshift=-.9cm] at (image.center) {\scriptsize $\pi^*$};
\node[yshift=0.5cm,xshift=.7cm] at (image.center) {\scriptsize $R_{\textrm{dn}}(\pi_0)$};
\node[yshift=.5cm,xshift=-.7cm] at (image.center) {\scriptsize $R_{\textrm{left}}(\pi_0)$};
\node[yshift=-1.45cm,xshift=-.6cm] at (image.center) {\scriptsize $R_{\textrm{up}}(\pi_0)$};
\node[yshift=-1.45cm,xshift=.7cm] at (image.center) {\scriptsize $R_{\textrm{right}}(\pi_0)$};
\end{tikzpicture}
\label{ex2_2}
}\hspace{2cm}
\subfigure[Inserting $(8,7),(14,3)$.]{
\begin{tikzpicture}
\node[anchor=south west,inner sep=0] (image) at (0,0) {\includegraphics[width=4cm]{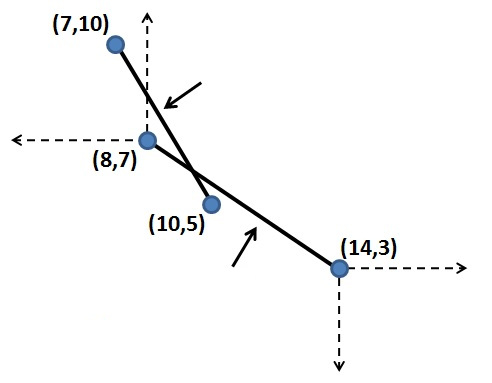}};
\node[yshift=.9cm,xshift=0cm] at (image.center) {\scriptsize $\pi_0.r$};
\node[yshift=-0.7cm,xshift=-.3cm] at (image.center) {\scriptsize $\pi^*$};
\node[yshift=0.2cm,xshift=1cm] at (image.center) {\scriptsize $R_{\textrm{dn}}(\pi^*)$};
\node[yshift=0.7cm,xshift=-1.4cm] at (image.center) {\scriptsize $R_{\textrm{left}}(\pi^*)$};
\node[yshift=-.9cm,xshift=-1.2cm] at (image.center) {\scriptsize $R_{\textrm{up}}(\pi^*)$};
\node[yshift=-.9cm,xshift=1.45cm] at (image.center) {\scriptsize $R_{\textrm{right}}(\pi^*)$};
\end{tikzpicture}
\label{ex4_1}
}
\caption{An example that shows the effect of the insertion of a segment on a set of non-dominated segments.}
\end{figure}

\begin{figure}
\centering
\subfigure[Tree after inserting (5,11).]{
\includegraphics[width=4cm]{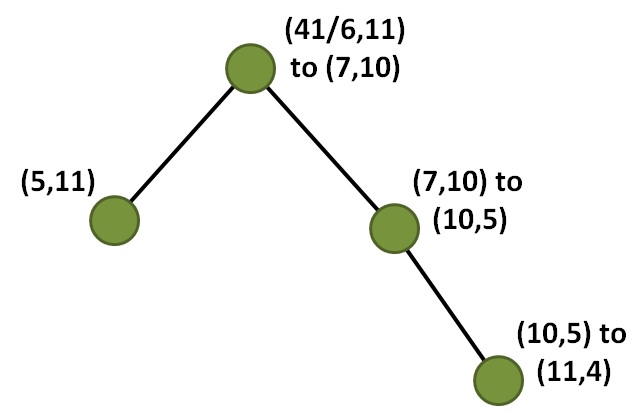} 
\label{ts1}
}\hspace{.5cm}
\subfigure[Tree after rebalancing.]{
\includegraphics[width=4cm]{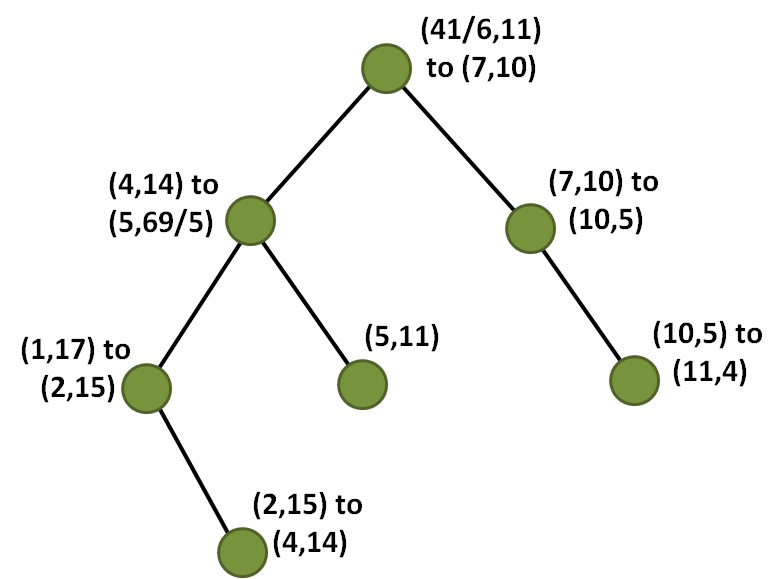} 
\label{ts3}
}
\subfigure[Final tree.]{
\includegraphics[width=4cm]{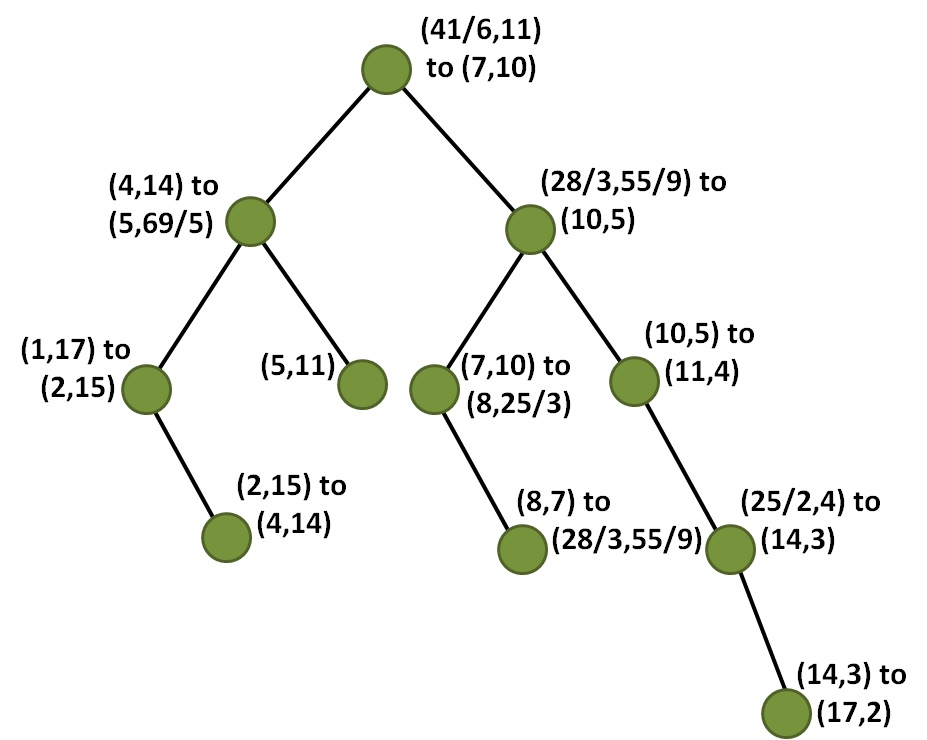} 
\label{ts6}
}
\caption{Rebalancing a tree upon insertion of point $(5,11)$.}
\label{ex_2}
\end{figure}

Next consider Pareto set (iv). Let $\pi^* \leftarrow (5,11)$ and call \textsc{Insert}$(\pi^*,\pi_0)$. Observe Figure \ref{ex2_2}. Because $\pi^*$ partially dominates $\pi_0$, we remove the dominated portion of $\pi_0$ by letting $S' = \pi_0.S \setminus cl(R_{\textrm{up}}(\pi^*))$. As a consequence, $\pi^* \subset R_{\textrm{left}}(S')$, and, since $\pi_0.l=\emptyset$, $\pi^*$ becomes the left child of $\pi_0$. Figure \ref{ts1} shows the BoT after $\pi^*$ has been inserted.
We leave it to the reader to consider Pareto set (ii). Rebalancing the subtree rooted at $\pi_0.l$ after processing this set yields the BoT shown in Figure \ref{ts3}. 

Consider now the insertion of Pareto set (v). Let $\pi^* \leftarrow [(8,7),(14,3)]$ and call \textsc{Insert}$(\pi^*,\pi_0)$. Clearly $\pi^* \subset R_{\textrm{right}}(\pi_0)$ and is hence inserted to $\pi_0.r$. Observe from Figure \ref{ex4_1}
that $\pi^*$ partially dominates $\pi_0.r$. This time, though, the portion of $\pi_0.r$ which is dominated is the center section of the segment. This means that $\pi_0.r$ must be split into two nodes $\pi_1$ and $\pi_2$. Node $\pi_1$ takes the place in the tree where $\pi_0.r$ originally was, and the left subtree of $\pi_0.r$ becomes the left subtree of $\pi_1$. Node $\pi_2$ becomes the right child of $\pi_1$ and the right subtree of $\pi_0.r$ becomes the right subtree of $\pi_2$. Subsequently, $\pi^* \subset R_{\textrm{right}}(\pi_1)$ and thus $\pi^*$ is inserted to $\pi_2$ (which is now $\pi_0.r.r$). Notice that $\pi_0.r.r$ partially dominates $\pi^*$ and that it is the center portion of $\pi^*$ that is dominated. Thus the calls to \textsc{Insert}$(\pi^*\cap R_{\textrm{left}}(\pi_0.r.r),\pi_0.r.r.l)$ and \textsc{Insert}$(\pi^*\cap R_{\textrm{right}}(\pi_0.r.r),\pi_0.r.r.r)$ each causes a portion of $\pi^*$ to be inserted at $\pi_0.r.r.l$ and $\pi_0.r.r.r$ respectively. Since $\pi_0.r.r.l = \emptyset$, $\pi^*\cap R_{\textrm{right}}(\pi_0.r.r)$ becomes $\pi.r.r.l$. Since $\pi_0.r.r.r$ is the segment (10,5) to (11,4), it is clear that another portion of $\pi^*$ needs to be removed, and then the remainder of $\pi^*$ becomes $\pi_0.r.r.r.r$. 

The remaining insertions are analogous to those that we have described. After yet another rebalance, the final BoT is as in Figure \ref{ts6}.

\section{Computational Experiments} \label{results}

We implemented the BoT in the C programming language and performed two tests. The first test addresses the efficiency with which a large number of randomly generated solutions can be stored in a BoT, using different rebalancing techniques. The second test addresses the utility a BoT when used in the two BB algorithms for BOMILP by \citet{belotti2012biobjective} and by \citet{adelgren2016}. All tests were run on Clemson University's Palmetto Cluster. Specifically, an HP SL250s server node with a single Intel E5-2665 CPU core with 16GB of RAM running Scientific Linux 6.4 was used.

In all of these experiments we compare the performance of the BoT with that of a dynamic linked list (L). Like a BoT, the linked list takes points and segments in $\R^2$ as input and stores only the nondominated subset of all input. All segments $S=(x_1,x_2,y_1,y_2)$ of the linked list are stored in increasing order of $x_1$, so that parsing all elements of the list produces the same output as an {\em in-order} visit of a BoT. Inserting a point or segment $S'$ consists in comparing it with every stored point or segment in the list, until a segment $S$ is encountered such that $S'\subset R_{\textrm{left}}(S)$. During each comparison, dominated solutions are discarded.  Although only a few, if any, elements of the list might be changed with the insertion of $S'$, because we do not know where these  elements are located and do not have more sophisticated search mechanisms on such a list, insertion has an {\em average} complexity of $O(t)$. Such lists have been used for storing nondominated solutions in both the pure integer \citep{sun1996} and mixed-integer cases \citep{mavrotas2005multi,vincent2013biobjective}.

Maintaining a balanced tree is one of the most costly operations, as shown in the tests below. As its only purpose is efficiency, it does not need to be applied at every step. Hence, we decided to further consider the rebalancing operations and use an alternative strategy that is less computationally costly, but still keeps the BoT fairly balanced.

We use the strategy of \citet{overmars1982}: for each non-leaf node $\pi$, the subtrees of $\pi.l$ and $\pi.r$  must contain no more than $\dfrac{1}{2-\delta}\textit{size}(\pi)$ nodes, where $\delta$ is a pre-selected value in the open interval $(0,1)$. This causes the depth of the tree to be at most $\log_{2-\delta} t$ where $t$ is the number of nodes in the tree. 

\citet{overmars1982} also suggest rebalancing by traversing the path travelled by an inserted solution in the reverse order and checking whether or not the balance criterion is satisfied at each of these nodes. This saves one from having to check the balance criterion at every node in the tree since the only places where it could have been altered are at nodes along this path. In our case, though, when a line segment is inserted into a BoT, it often does not remain intact, but may be separated into many smaller segments, each traversing its own path through the tree before finally being added. For this reason, rearranging the tree after insertion is troublesome, and we experimented with a few alternative approaches:
\begin{enumerate}
\item[\textbf{A0} -] No rebalancing is used.
\item[\textbf{A1} -] Before inserting a point or segment at the root node, check the balance criterion at every node in the tree and rebalance where necessary. This approach guarantees that the balance of the tree is maintained, though at a high computational cost.
\item[\textbf{A2} -] Check the balance criterion after the $k$-th insertion (we used $k=100$ in our tests), then check the balance criterion when the tree size increased by $101\%$ w.r.t.\mbox{} the size at the previous check. This approach significantly decreases the complexity of rebalancing, but eliminates the balance guarantee.
\item[\textbf{A3} -] Check the balance criterion at any node that is currently being inserted at. This approach has a much lower complexity than A1, and would cause balance to be maintained at the root node, and along any frequently travelled paths in the tree. However, again the guarantee of balance is lost.
\item[\textbf{A4} -] Combine approaches A2 and A3: check the balance criterion of the entire tree after the $k$-th insertion ($k = 100$), then check the balance criterion again when the tree size increases by 800\% w.r.t. the size at the previous check. In between these checks of the entire tree, check the balance criterion for any node that is being inserted at.

Approach A4 allows for maintaining a fairly well balanced tree by applying approach A2 much more infrequently than if using approach A2 alone. Clearly this has a higher complexity than approach A3, but it may be less than that of approach A2 and allow for a more balanced tree.
\end{enumerate} 

We implemented each of these approaches in our first experiment, described in Section \ref{sec2.2.1}.
We utilize approach A2 when performing our second experiment, which is described in Section \ref{tests-bb} because for most of our tests A2 performed comparably to A0 in terms of CPU time, but always maintained a more balanced tree.

\subsection{Insertion of large number of random points}
\label{sec2.2.1}

This first test has two main purposes: (i) to compare the efficiency of a BoT with that of a dynamic list when storing nondominated solutions, and (ii) to determine the best rebalancing approach w.r.t.\mbox{} tree depth and time.

The test consists of repeating the following procedure until $N$ insertions have been made into a BoT or the dynamic list. First, generate a random integer $i \in [1,6]$ and a random number $r_1 \in (0,10)$. Then, if $i>1$, for each $j\in\{2,\dots,i\}$ a random number $c_j\in (0,1)$ is generated and we define $r_j = r_1 + \su[\ell=2]{j}c_\ell$. Next, for each $j\in\{1,\dots,i\}$ the following are computed: (i) $y_j = \dfrac{(10.5-r_j)^2}{5} -k$, and (ii) $x_j = r_j + (5-k)$.
Here $k$ is a dynamic value which is defined as 1 at the start of the test and increases by $\dfrac{\mu}{N}$ each time the above process is repeated, and $\mu \in \R$ determines how much the solutions should ``improve'' over the course of the test. If $i=1$, the singleton $(x_1,y_1)$ is inserted into the structure, otherwise the points $(x_1,y_1),...,(x_i,y_i)$ are arranged in order of increasing $x$ values and then the line segments connecting each adjacent pair of points are inserted into the structure. We performed this test 100 times for each combination of the values $N\in \{10^4,10^5,10^6,10^7\}$ and $\mu\in\{0,0.001,0.01,0.1,1,10\}$. We used various values for $\delta$ and found that the results were quite similar, but determined to use a value of $\delta = 0.3$. For each test we recorded the total insertion time for both algorithms, the final depth of a BoT, and the final number of nodes stored in the BoT and in the list. 

With $\mu$ close to zero, there is little or no separation between early generated solutions and later generated ones, and all are likely to be Pareto. With large values of $\mu$, there is significant separation between early generated solutions and later generated ones, the latter being much more likely to be Pareto. 
Figure \ref{ex_data_ex1} shows an example of solutions generated during this experiment for $\mu=0.1,1$ and $10$ and for $N=100$. The red solutions are those that are in the BoT at the end of the test. 

\begin{figure}
\subfigure[$\mu=0.1$]{\includegraphics[width=5cm]{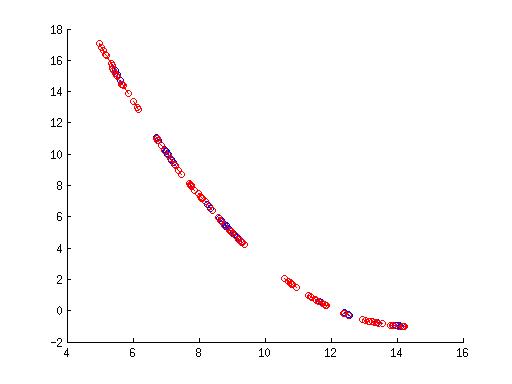}}
\subfigure[$\mu=1$]{\includegraphics[width=5cm]{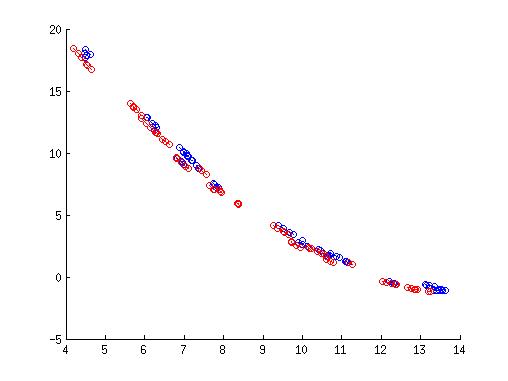}}
\subfigure[$\mu=10$]{\includegraphics[width=5cm]{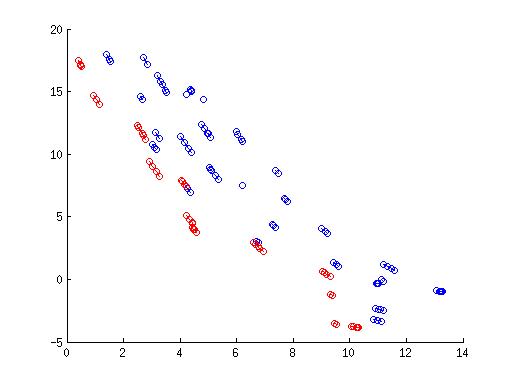}}
\caption{Example of randomly generated solutions; $N=100$.}\label{ex_data_ex1}
\end{figure}





The minimum, maximum, and geometric means of the CPU times and final depths of the tree can be found in Tables \ref{table_results_a}-\ref{table_results_c}. 
The symbols A0 -- A4 and L indicate runs in which the various rebalancing approaches of a BoT and the dynamic list were used for storing solutions. Also, entries in Tables 
\ref{table_results_a}-\ref{table_results_c}
which contain dashes are those for which no results are available due to the fact that individual runs took over 12 hours to complete and were therefore terminated. The symbols $\circledast$ and $\dagger$ indicates results for which, due to the large amount of time taken for each individual run, we were unable to perform the test 100 times. For these results, each test was instead run 5 times (for $\circledast$) and 3 (for $\dagger$).

We use performance profiles \citep{dolan2002benchmarking} to show the relative effectiveness of the various rebalancing approaches, in terms of CPU time and tree depth, in Figure \ref{exp1_profs}. We omit the results of the list implementation and rebalancing approach A1 as they performed poorly in terms of CPU time when compared to the other approaches, and only show profiles for $N=10^7$ and for a few values of $\mu$ for reasons of space. For $\mu=0$ (Figure \ref{perf_prof1}), A2, shown by the curve with $+$ symbols, has the best time performance although it does not dominate A0 (the solid line), which in turn is at least 30\% worse than A2 in half of the instances but never worse than 40\%. Algorithms A3 and A4 perform similarly and are dominated by A0 and A2. For $\mu=0.001$, instead (see Figure \ref{perf_prof2}), A0 dominates A2, which is, nevertheless, never more than 15\% worse than the best performance. A3 and A4 (lines with $*$ and $\circ$, respectively) perform worse still, and this is confirmed for $\mu=1$ (see Figure \ref{perf_prof3}), where A4 fares slightly better than A3. This is confirmed for $\mu=10$ (not shown).

While A2 and A0 perform well in general, A0 carries the risk of an unbalanced tree, as reported in Figure \ref{perf_prof4}, where tree depth is shown to be up to 3.5 times worse than the best. A3 and A4 have the best performance in terms of maximum tree depth, which however reflects in poor time performance as shown in Figure \ref{perf_prof2}. 

\begin{figure}
\subfigure[\small Time, $\mu = 0$]{
\includegraphics[width=.45\linewidth]{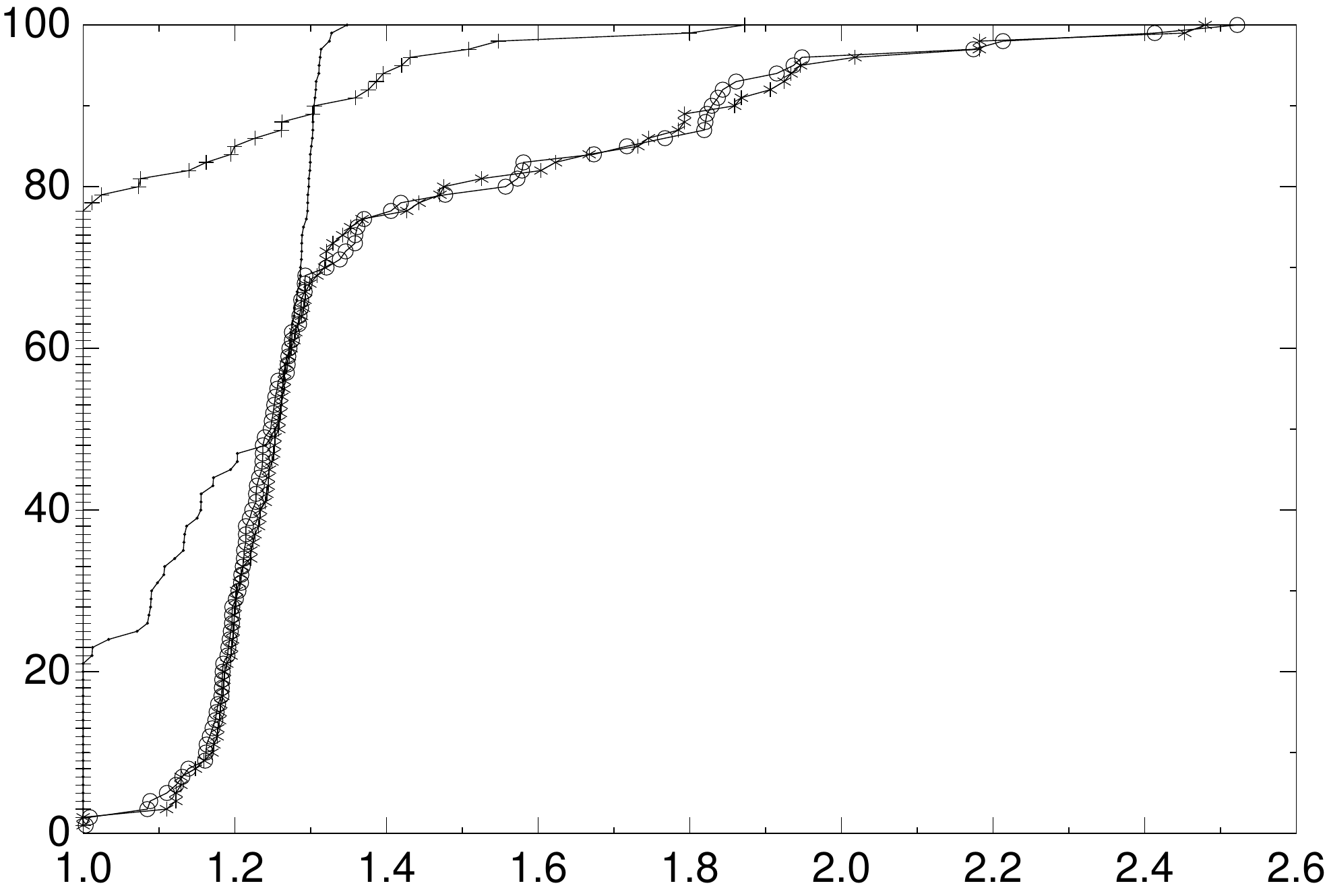}
\label{perf_prof1}
}\hfill
\subfigure[\small Time, $\mu=0.001$]{
\includegraphics[width=.45\linewidth]{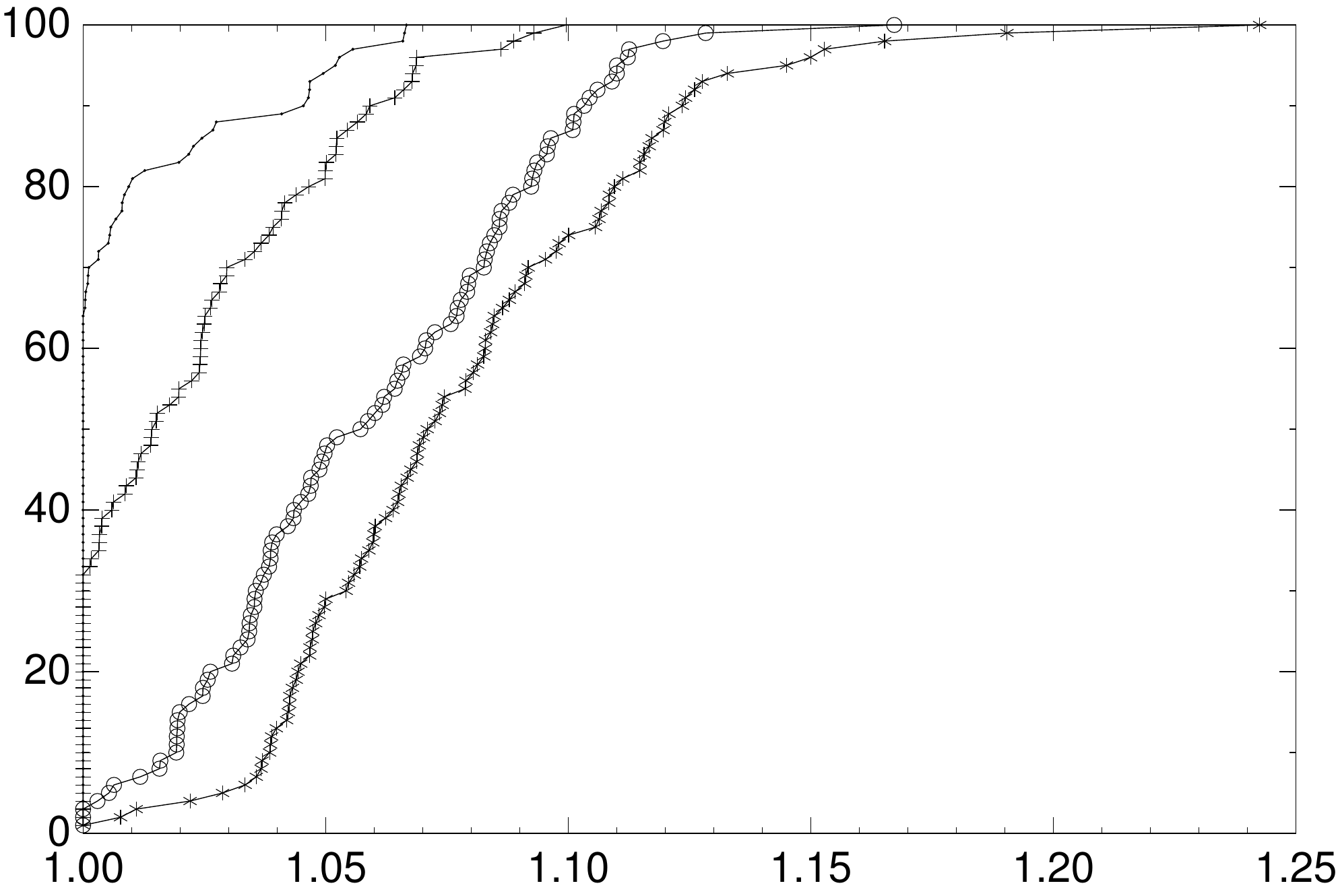}
\label{perf_prof2}
}
\subfigure[\small Time, $\mu=1$]{
\includegraphics[width=.45\linewidth]{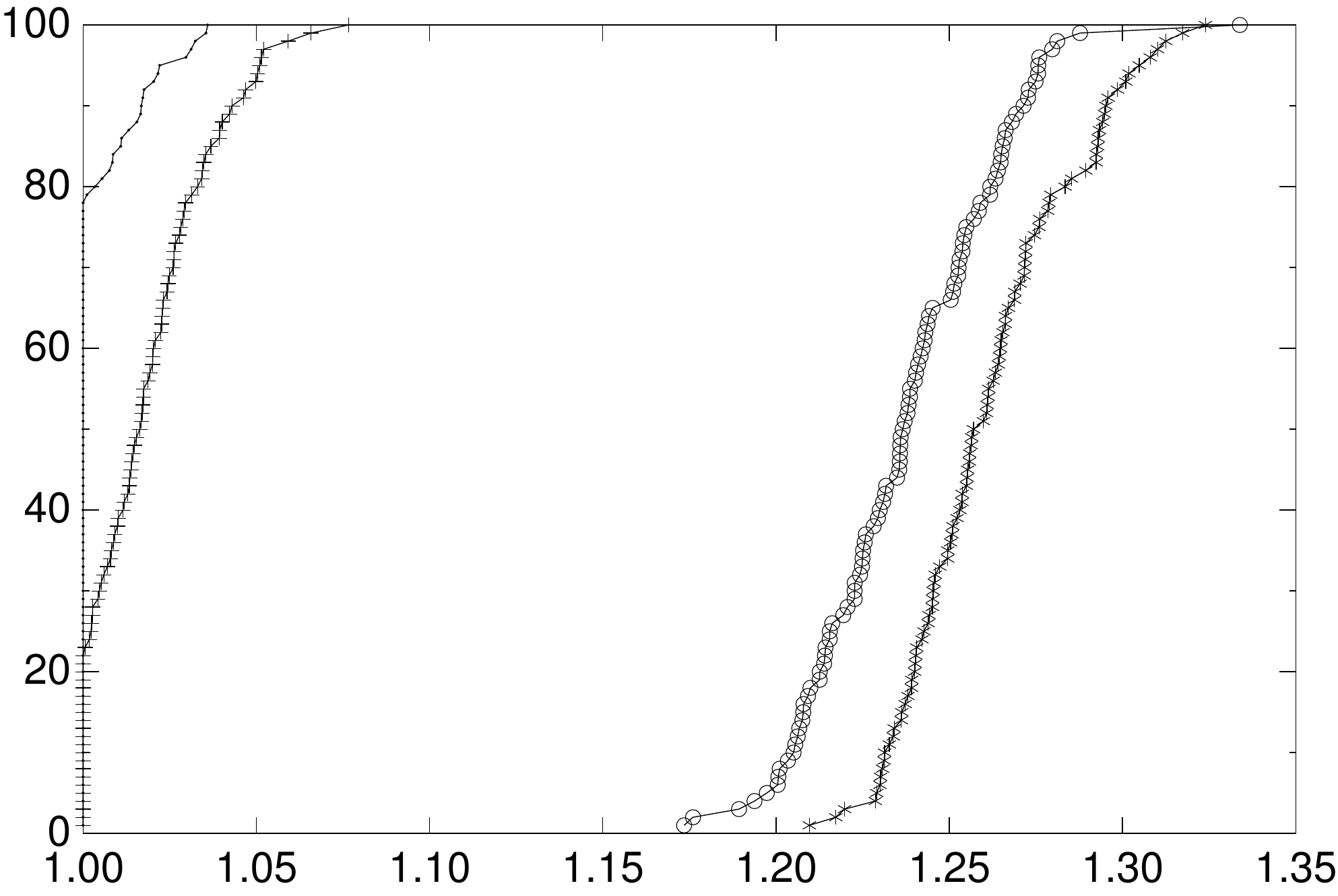}
\label{perf_prof3}
}\hfill
\subfigure[\small Depth,  $\mu = 0.001$]{
\includegraphics[width=.45\linewidth]{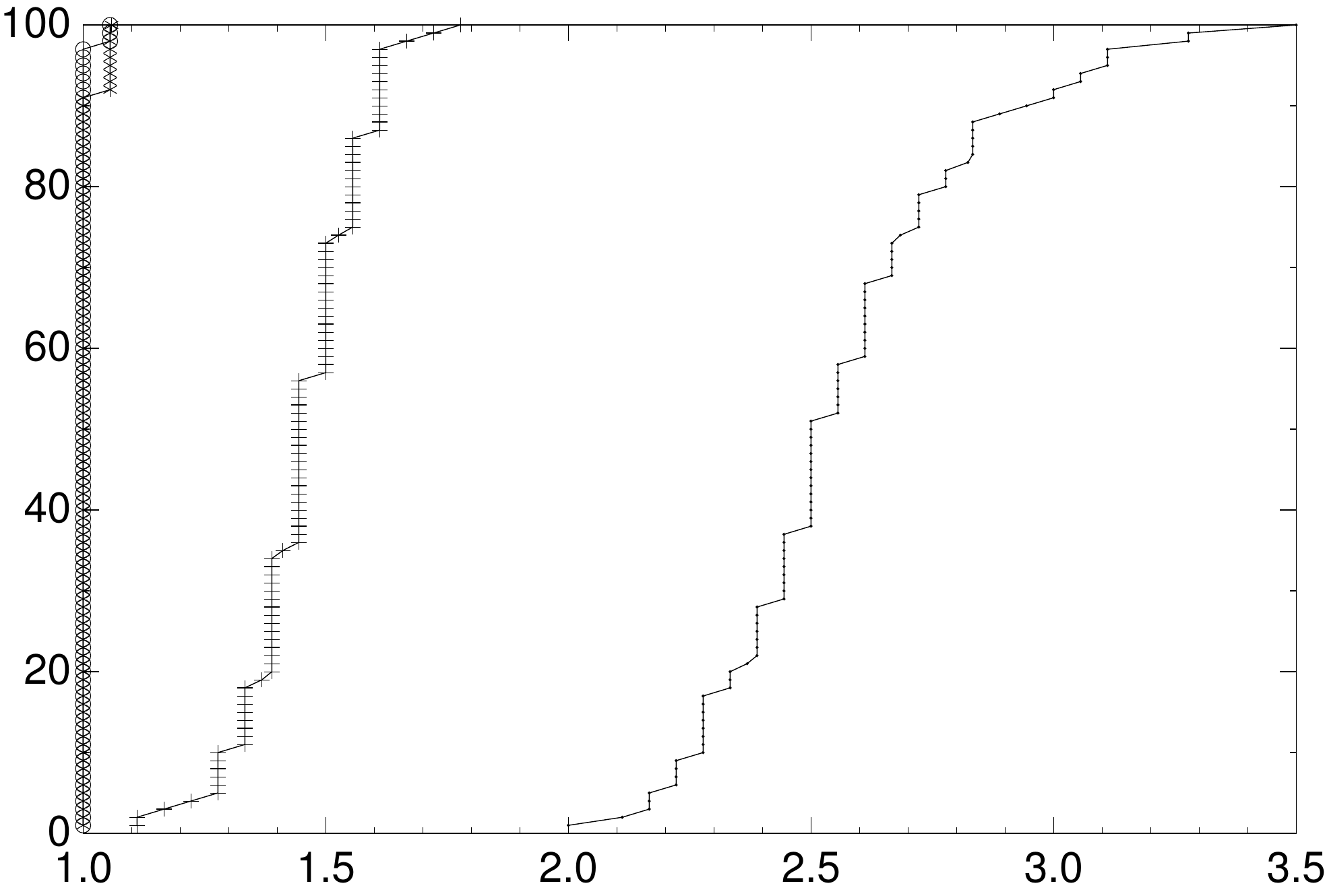}
\label{perf_prof4}
}
\caption{Performance profiles of the rebalancing algorithms A0 (solid line), A2 ($+$), A3 ($*$), and A4 ($\circ$). All profiles are w.r.t.\mbox{} results for $N=10^7$. For each algorithm $A$, a point $(x,y)$ on $A$'s curve indicates that for $y$\% of the instances $A$'s performance is at most $x$ times worse than the best algorithm.}\label{exp1_profs}
\end{figure}

\begin{table}[t]
\begin{center}
\scriptsize
\caption{
Time and depth of the tree for the random points test ($\mu = 0, 0.001$)}
\label{table_results_a}
\begin{tabular}{rl|rrr|rrr||rrr|rrr}
\toprule
& Rebal & 
\multicolumn{3}{c|}{Time (s)} & \multicolumn{3}{c||}{Depth}&
\multicolumn{3}{c|}{Time (s)} & \multicolumn{3}{c}{Depth} \\ 
$N$ & Type & 
Min	& Avg	& Max	& Min	& Avg	& Max & 
Min	& Avg	& Max	& Min	& Avg	& Max \\
\hline 
&&\multicolumn{6}{c||}{$\mu = 0$} &
\multicolumn{6}{c}{$\mu = 0.001$}\\
\hline 
	$10^4$      &	A0	&	0.08	&	0.09	&	0.10    &	34	&	42.6	&	53	&	0.03	&	0.03	&	0.04	&	25	&	31.6	&	41	\\
		    &	A1	&	2.22	&	2.31	&	2.55	&	16	&	16.2	&	18	&	0.28	&	0.30	&	0.36	&	12	&	12.0	&	13	\\
		    &	A2	&	0.16	&	0.20	&	0.35	&	16	&	17.1	&	18	&	0.05	&	0.05	&	0.08	&	12	&	13.6	&	16	\\
		    &	A3	&	0.11	&	0.16	&	0.38	&	16	&	17.0	&	18	&	0.04	&	0.05	&	0.1	&	12	&	13.1	&	14	\\
		    &	A4	&	0.11	&	0.16	&	0.39	&	16	&	17.0	&	18	&	0.04	&	0.05	&	0.09	&	12	&	13.1	&	14	\\
		    &	L	&	5.65	&	5.94	&	6.27	&	--	&	--	&	--	&	0.35	&	0.37	&	0.41	&	--	&	--	&	--	\\
\hline
	$10^5$      &	A0	&	4.52	&	4.85	&	4.94	&	42	&	51.3	&	62	&	0.64	&	0.70	&	0.79	&	27	&	36.7	&	52	\\
		    &	A1	&	448.54	&	476.75	&	520.94	&	19	&	19.2	&	20	&	9.09	&	9.57	&	11.02	&	13	&	14.0	&	15	\\
		    &	A2	&	6.31	&	 9.91	&	32.54	&	19	&	20.2	&	22	&	0.80	&	0.91	&	1.12	&	14	&	17.0	&	19	\\
		    &	A3	&	5.83	&	 9.81	&	41.62	&	19	&	19.8	&	21	&	0.84	&	0.99	&	1.30	&	14	&	14.9	&	16	\\
		    &	A4	&	5.69	&	 9.77	&	43.83	&	19	&	19.7    &	21	&	0.81	&	0.99	&	1.31	&	14	&	14.9	&	16	\\
		    &	L	&	668.19	&	694.63	&	747.21	&	--	&	--	&	--	&	19.30	&	19.65	&	20.73	&	--	&	--	&	--	\\
\hline
	$10^6$      &	A0	&	126.01	&	128.91	&	135.84	&	45	&	54.6	&	70	&	17.32	&	18.39	&	18.79	&	34	&	42.2	&	54	\\
		    &	A1	&	--	&	--	&	--	&	--	&	--	&	--	&	427.00	&	434.40	&	465.09	&	17	&	17	&	17	\\
		    &	A2	&	124.02	&	143.02	&	254.55	&	21	&	22.4	&	24	&	18.93	&	19.82	&	21.71	&	17	&	21.0	&	24	\\
		    &	A3	&	137.08	&	168.69	&	328.82	&	21	&	21.0	&	22	&	20.12	&	21.35	&	24.55	&	16	&	16.3	&	17	\\
		    &	A4	&	136.23	&	165.19	&	306.11	&	20	&	21.0	&	22	&	19.85	&	21.07	&	24.2	&	16	&	16.44	&	17	\\
		    &	L	&	--	&	--	&	--	&	--	&	--	&	--	&	--	&	--	&	--	&	--	&	--	&	--	\\
\hline
	$10^7$      &	A0	&	1684.59	&	2122.15	&	2809.04	&	50	&	57.9	&	68	&	431.54	&	459.64	&	485.31	&	36	&	46.3	&	63	\\
		    &	A1	&	--	&	--	&	--	&	--	&	--	&	--	&$\circledast$39654.42&	39922.79&	40136.91&	17	&	17.0	&	17	\\
		    &	A2	&	1287.45	&	1915.42	&	4908.04	&	22	&	24.4	&	26	&	423.87	&	466.23	&	504.79	&	20	&	26.2	&	32	\\
		    &	A3	&	1599.44	&	2422.43	&	6690.78	&	21	&	21.3	&	23	&	468.92	&	491.24	&	536.21	&	18	&	18.1	&	19	\\
		    &	A4	&	1611.95	&	2413.61	&	6880.00	&	21	&	21.3	&	23	&	456.67	&	482.73	&	522.69	&	17	&	18.0	&	19	\\
		    &	L	&	--	&	--	&	--	&	--	&	--	&	--	&$\dagger$63936.51&	67937.70&	70214.80&	--	&	--	&	--	\\
\bottomrule
\end{tabular}
\end{center}
\end{table}

\begin{table}[t]
\begin{center}
\scriptsize
\caption{
Time and depth of the tree for the random points test ($\mu = 0.01, 0.1$)}
\label{table_results_b}
\begin{tabular}{rl|rrr|rrr||rrr|rrr}
\toprule
& Rebal & 
\multicolumn{3}{c|}{Time (s)} & \multicolumn{3}{c||}{Depth}&
\multicolumn{3}{c|}{Time (s)} & \multicolumn{3}{c}{Depth} \\ 
$N$ & Type & 
Min	& Avg	& Max	& Min	& Avg	& Max & 
Min	& Avg	& Max	& Min	& Avg	& Max \\
\hline 
&&\multicolumn{6}{c||}{$\mu = 0.01$} &
\multicolumn{6}{c}{$\mu = 0.1$}\\
\hline 
	$10^4$	&	A0	&	0.02	&	0.02	&	0.02	&	14	&	21.6	&	36	&	0.01	&	0.01	&	0.02	&	11	&	12.8	&	16	\\
		&	A1	&	0.11	&	0.11	&	0.13	&	10	&	10.2	&	11	&	0.05	&	0.05	&	0.05	&	8	&	8.7	&	10	\\
		&	A2	&	0.03	&	0.03	&	0.03	&	10	&	12.2	&	16	&	0.02	&	0.02	&	0.02	&	8	&	10.5	&	14	\\
		&	A3	&	0.03	&	0.03	&	0.03	&	10	&	11.1	&	12	&	0.02	&	0.02	&	0.02	&	9	&	9.4  	&	11	\\
		&	A4	&	0.03	&	0.03	&	0.03	&	10	&	11.1	&	12	&	0.02	&	0.02	&	0.02	&	9	&	9.4  	&	11	\\
		&	L	&	0.11	&	0.12	&	0.14	&	--	&	--	    &	--	&	0.04	&	0.04	&	0.06	&	--	&	--	&	--	\\
\hline
	$10^5$	&	A0	&	0.34	&	0.36	&	0.39	&	18	&	25.4	&	47	&	0.23	&	0.23	&	0.24	&	13	&	15.1	&	21	\\
		&	A1	&	3.14	&	3.22	&	3.35	&	12	&	12.1	&	13	&	1.18	&	1.20	&	1.30	&	10	&	10.1	&	12	\\
		&	A2	&	0.41	&	0.43	&	0.47	&	13	&	15.9	&	20	&	0.26	&	0.26	&	0.27	&	10	&	14.2	&	18	\\
		&	A3	&	0.47	&	0.50	&	0.55	&	12	&	13.0	&	14	&	0.33	&	0.34	&	0.36	&	10	&	11.1	&	12	\\
		&	A4	&	0.48	&	0.50	&	0.56	&	12	&	13.0	&	14	&	0.34	&	0.34	&	0.37	&	10	&	11.1	&	12	\\
		&	L	&	4.12	&	4.27	&	4.47	&	--	&	--	    &	--	&	1.25	&	1.31	&	1.37	&	--	&	--	&	--	\\
\hline
	$10^6$	&	A0	&	6.61	&	7.07	&	7.88	&	21	&	29.6	&	43	&	3.50	&	3.58	&	3.76	&	16	&	18.0	&	21	\\
		&	A1	&	91.56	&	96.47	&	105.81	&	13	&	13.9	&	15	&	32.01	&	32.41	&	32.85	&	12	&	12.1	&	13	\\
		&	A2	&	6.92	&	7.37	&	7.84	&	16	&	19.8	&	25	&	3.69	&	3.78	&	3.92	&	14	&	17.7	&	20	\\
		&	A3	&	8.31	&	8.83	&	9.56	&	14	&	14.9	&	16	&	4.89	&	5.03	&	5.25	&	12	&	13.0	&	14	\\
		&	A4	&	8.24	&	8.72	&	9.14	&	14	&	14.9	&	16	&	4.90	&	5.04	&	5.21	&	12	&	13	&	14	\\
		&	L	&	211.81	&	215.19	&	224.98	&	--	&	--	    &	--	&	--	&	--    	&	--   	&	--	&	--	&	--	\\
\hline
	$10^7$	&	A0	&	94.29	&	184.30	&	188.04	&	23	&	33.2	&	48	&	68.06	&	70.71	&	74.91	&	19	&	20.9	&	24	\\
		&	A1	&	1874.98	&	4396.94	&	6018.77	&	15	&	15.1	&	16	&	948.30	&	1027.48	&	1132.85	&	13	&	13.9	&	15	\\
		&	A2	&	93.15	&	187.70	&	191.03	&	18	&	23.5	&	29	&	69.08	&	71.76	&	74.85	&	16	&	20.7	&	23	\\
		&	A3	&	108.64	&	206.37	&	210.70	&	16	&	16.3	&	17	&	85.63	&	88.87	&	93.43	&	14	&	14.8	&	15	\\
		&	A4	&	109.23	&	200.11	&	207.89	&	16	&	16.3	&	17	&	84.15	&	87.15	&	92.13	&	14	&	14.8	&	15	\\
		&	L	&	2786.57	&	14687.80 &	15234.72&	--	&	--	&	--	&	2136.03	&	2210.22	&	2452.43	&	--	&	--	&	--	\\
\bottomrule
\end{tabular}
\end{center}
\end{table}

\begin{table}[t]
\begin{center}
\scriptsize
\caption{
Time and depth of the tree for the random points test ($\mu = 1, 10$)}
\label{table_results_c}
\begin{tabular}{rl|rrr|rrr||rrr|rrr}
\toprule
& Rebal & 
\multicolumn{3}{c|}{Time (s)} & \multicolumn{3}{c||}{Depth}&
\multicolumn{3}{c|}{Time (s)} & \multicolumn{3}{c}{Depth} \\ 
$N$ & Type & 
Min	& Avg	& Max	& Min	& Avg	& Max & 
Min	& Avg	& Max	& Min	& Avg	& Max \\
\hline 
&&\multicolumn{6}{c||}{$\mu = 1$} &
\multicolumn{6}{c}{$\mu = 10$}\\
\hline 
	$10^4$	    &	A0	&	0.01	&	0.01	&	0.01	&	9	&	11.2	&	16	&	0.01	&	0.01	&	0.01	&	9	&	10.9	&	17	\\
		    &	A1	&	0.03	&	0.04	&	0.04	&	7	&	8.0	&	9	&	0.03	&	0.03	&	0.03	&	7	&	7.4	&	9	\\
		    &	A2	&	0.01	&	0.01	&	0.02	&	8	&	9.8	&	13	&	0.01	&	0.01	&	0.01	&	7	&	8.9	&	12	\\
		    &	A3	&	0.01	&	0.02	&	0.02	&	8	&	8.6	&	10	&	0.01	&	0.01	&	0.01	&	7	&	8.0	&	9	\\
		    &	A4	&	0.01	&	0.02	&	0.02	&	8	&	8.6	&	10	&	0.01	&	0.01	&	0.01	&	7	&	8.0 	&	9	\\
		    &	L	&	0.03	&	0.03	&	0.04	&	--	&	--	&	--	&	0.03	&	0.03	&	0.04	&	--	&	--	&	--	\\
\hline
	$10^5$	    &	A0	&	0.19	&	0.19	&	0.19	&	11	&	12.6	&	17	&	0.16	&	0.16	&	0.17	&	12	&	16.7	&	24	\\
		    &	A1	&	0.55	&	0.56	&	0.61	&	8	&	8.8   	&	10	&	1.31	&	1.40	&	1.49	&	8	&	8.53	&	10	\\
		    &	A2	&	0.20    &	0.20	&	0.21	&	9	&	12.3	&	17	&	0.23	&	0.23	&	0.25	&	9	&	11.1	&	15	\\
		    &	A3	&	0.25	&	0.25	&	0.26	&	9	&	9.3	&	10	&	0.28	&	0.29	&	0.31	&	8	&	9.0  	&	10	\\
		    &	A4	&	0.25	&	0.25	&	0.26	&	9	&	9.3	&	10	&	0.28	&	0.29	&	0.31	&	8	&	9.0 	&	10	\\
		    &	L	&	0.50    &	0.52	&	0.57	&	--	&	--	&	--	&	2.51	&	2.67	&	2.91	&	--	&	--	&	--	\\
\hline
	$10^6$	    &	A0	&	2.39	&	2.42	&	2.47	&	14	&	15.3	&	19	&	1.92	&	1.94	&	1.96	&	11	&	12.6	&	20	\\
		    &	A1	&	12.15	&	12.44	&	12.85	&	10	&	10.2	&	11	&	5.62	&	5.72	&	5.89	&	8	&	8.7 	&	10	\\
		    &	A2	&	2.47	&	2.50	&	2.57	&	13	&	15.2	&	19	&	1.96	&	1.98	&	2.01	&	10	&	12.5	&	20	\\
		    &	A3	&	3.48	&	3.54	&	3.64	&	10	&	11.1	&	12	&	2.53	&	2.56	&	2.61	&	9	&	9.41	&	11	\\
		    &	A4	&	3.48	&	3.55	&	3.63	&	10	&	11.1	&	12	&	2.54	&	2.57	&	2.62	&	9	&	9.41	&	11	\\
		    &	L	&	13.18	&	13.44	&	13.87	&	--	&	--	&	--	&	5.20	&	5.38	&	5.54	&	--	&	--	&	--	\\
\hline
	$10^7$	    &	A0	&	34.49	&	34.97	&	35.78	&	16	&	18.0	&	21	&	23.28	&	23.97	&	24.25	&	13	&	15.23	&	19	\\
		    &	A1	&	317.90	&	321.25	&	329.82	&	12	&	12.1	&	13	&	121.22	&	124.55	&	126.14	&	10	&	10.1	&	12	\\
		    &	A2	&	35.15	&	36.10	&	37.20	&	16	&	18.0	&	21	&	24.24	&	24.45	&	24.87	&	13	&	15.2	&	19	\\
		    &	A3	&	48.55	&	49.79	&	52.08	&	12	&	13.0	&	14	&	34.57	&	35.44	&	36.21	&	10	&	11.2	&	13	\\
		    &	A4	&	48.54	&	49.84	&	51.46	&	12	&	13.0	&	14	&	35.22	&	35.56	&	36.78	&	10	&	11.2	&	13	\\
		    &	L	&	419.88	&	435.56	&	458.88	&	--	&	--	&	--	&	131.85	&	134.998	&	137.2	&	--	&	--	&	--	\\
\bottomrule
\end{tabular}
\end{center}
\end{table}

As shown in Tables \ref{table_results_a}--\ref{table_results_c}, with approaches A0, A2, A3, and A4, the BoT is able to process inserted solutions much more quickly than the dynamic list; A1 is more efficient than the list, but far slower than the other approaches. 
The performance difference grows with $N$.

Also, for most values of $N$ and $\mu$, A0 typically performs the best in terms of running time, followed closely by A2. On the other hand, A3 and A4 perform the best in terms of tree depth. Because an increasing $\mu$ implies an increasing number of eliminated solutions at each insertion (and, in general, fewer solutions stored in the data structure at any time), the time taken to insert solutions decreases as the value of $\mu$ increases. Furthermore, the larger the value of $\mu$, the smaller the gap between the CPU time spent by the linked list and that the BoT. 
We conclude that the BoT is more scalable and performs better than a linked list for storing randomly generated solutions. Of all rebalancing algorithms tested, only A1 appears to be inefficient, while no rebalancing at all (algorithm A0) seems to have a limited impact on the performance. The results suggest that A2 seems to combine good performance in terms of time and maximum tree depth, and is therefore the rebalancing algorithm of choice for our next tests.


\begin{figure}
\centering
\includegraphics[width=.70\textwidth]{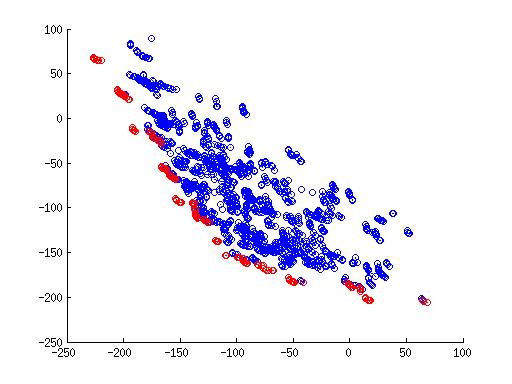}
\caption{Set of all solutions encountered by the BB method
in \citet{belotti2012biobjective} with 80 variables and 80 constraints. The red dots are the Pareto optimal solutions.
}\label{warm_start}
\end{figure}


\subsection{Using a BoT in branch-and-bound algorithms}
\label{tests-bb}

To test more in depth the utility of a BoT when utilized in a BB, we performed another set of tests in which a variety of instances of BOMILP were solved using branch-and-bound. To ensure the results were not too dependent on one solver, we ran our tests on two implementations, described in \citet{belotti2012biobjective} and \citet{adelgren2016}. The latter utilizes this data structure not only for storing found solutions, but also to check for domination of bound sets and hence, for fathoming. In addition, it uses {\em warm starting}, i.e., it generates an initial set of points before the BB is started. This BB was able to solve larger instances in under 8 hours and also did not run into numerical difficulties when solving instances of type II from \citet{boland2013biobjective}. For both BB solvers, we have used approach A2 for rebalancing.


We used three classes of BOMILP instances: two sets of instances described and used in \citet{belotti2012biobjective} and \citet{boland2013biobjective} and one that we have created. 
%
%
The latter set contains instances, which we call {\em geometrical}, that, while of rather simple structure, contain by design a large number of Pareto points. Testing these instances forces the BB algorithms to create numerous  BB nodes and solutions, thus  assessing a BoT in a more realistic, large-scale setting where the number of solutions inserted are closer to those of our first experiment. These are very simple BOMILPs with a large Pareto set for which a single BB node is solved very quickly. The model is as follows:
\begin{equation}
  \label{eq:newinst}
  \begin{array}{llrrr}
    \max & x \\
    \max & y \\
    \textrm{s.t.} & y - s_i x &\leq &  b_i - a_i s_i & \quad \forall i \in \{1,\dots,k\}\\
                  & x - 0.1 z &\leq &  0.025 \\
                  & x - 0.1 z &\geq & -0.025 \\
                  & \multicolumn{3}{l}{x,y \geq 0, z \in \Z,}
  \end{array}
\end{equation}
where $a_1 = 0$, $a_{i+1} = a_i + \frac{N}{k}$ for $i \in \{1,\dots,k\}$, $b_i = \sqrt{N^2-a_i^2}$ for $i \in \{1,\dots,k+1\}$, $s_i = \frac{b_{i+1}-b_i}{a_{i+1}-a_i}$ for $i \in \{1,\dots,k\}$, and $N$ and $k$ are parameters.



The results from the branch-and-bound tests are given in Tables \ref{table:BSWold} to \ref{table:AGgraph}. We use T and L to represent the implementations using a BoT and the dynamic list, respectively. From \citet{belotti2012biobjective}, 30 instances are available for each problem size and 5 (resp.\mbox{} 4) from \citet{boland2013biobjective} for each problem size of Type I (resp.\mbox{} II). The BB by 
\citet{belotti2012biobjective} works by creating $M$ initial Pareto points via solving $M$ single-objective MIPs, hence we tested using different values of $M$. In Tables \ref{table:BSWold} and \ref{table:AGold}, we report geometric means of the results on instances from \citet{belotti2012biobjective}, and provide results for every instance from \citet{boland2013biobjective} (with geometric means reported in bold). We report the total solve time, the time spent inserting solutions (Time (insert)), the time spent doing other data structure operations such as fathoming (Time (other), only reported for the BB by \citet{adelgren2016} as it is negligible otherwise), the number of BB nodes, the number of insertions, the number of nodes in the structure at termination (\#nodes) and the depth of the BoT. Again, problem sizes are reported as the number of variables, which in most cases also equals the number of constraints (for instances from \citet{boland2013biobjective} of Type II, the number of constraints is slightly greater than the number of variables).

In principle, a BB algorithm is expected to follow the same path regardless of the data structure, with the same number of BB nodes explored and of solutions stored. While this is true for most tests, there were discrepancies when using the BB by \citet{adelgren2016} for solving instances from \citet{belotti2012biobjective} and \citet{boland2013biobjective}. We show the number of BB nodes, of insertions, and of nodes in Table \ref{table:AGold} when these differ, in general by a small amount but up to 3\% in a few cases. Extra tests suggest that these discrepancies are due to small differences in the solutions contained in either data structure: the insertion of a solution $\pi$ into a BoT and a list that have, initially, the same set of solutions may result in different sets if, for instance, the extremes of segment $\pi.S$ are very close to previous solutions, say within a tolerance of $10^{-4}$, since the operations of Algorithm \ref{insert} allow for multiple changes in a subtree of the BoT. This effect can be amplified as it may result in one BB node to be fathomed or explored, the latter resulting in an entire new BB subtree being explored.

\begin{table}
\begin{center}
{
\caption{
BB tests: results for algorithm of \citet{belotti2012biobjective}. 
} \label{table:BSWold}
}
{\scriptsize
\begin{tabular}{rr|rr|rr|r|r|r|r}
\toprule
     &     & \multicolumn{2}{c|}{Time}                      & \multicolumn{2}{c|}{Time (insert)}              & \multicolumn{1}{c|}{\#BBnodes}        & \multicolumn{1}{c|}{\#inserts}       & \multicolumn{1}{c|}{\#nodes} & tree\\ 
Size & $M$ & \multicolumn{1}{c}{T} & \multicolumn{1}{c|}{L} & \multicolumn{1}{c}{T} & \multicolumn{1}{c|}{L} & \multicolumn{1}{c|}{ $(\times 10^2)$} & \multicolumn{1}{c|}{$(\times 10^3)$} & \multicolumn{1}{c|}{}        & depth\\  
\midrule
\multicolumn{10}{c}{Instances from \citet{belotti2012biobjective} (averaged)}\\
\midrule
60	&	10	&	15.0	&	15.1	&	0.007	&	0.007	&	8.11	&	1.790	&	62.6	&	6.3	\\
	&	25	&	11.6	&	11.6	&	0.005	&	0.004	&	6.32	&	1.155	&	60.0	&	6.6	\\
	&	50	&	16.4	&	16.1	&	0.012	&	0.007	&	8.22	&	1.838	&	64.1	&	6.4	\\
80	&	10	&	38.0	&	38.2	&	0.011	&	0.009	&	13.55	&	2.729	&	82.9	&	6.9	\\
	&	25	&	32.6	&	32.9	&	0.005	&	0.006	&	11.42	&	2.046	&	82.1	&	7.0	\\
	&	50	&	42.6	&	42.9	&	0.011	&	0.008	&	14.11	&	2.988	&	80.6	&	6.9	\\
\midrule
\multicolumn{10}{c}{Instances from \citet{boland2013biobjective}, type 1}\\
\midrule
80	&	50	&	101	&	102	&	0.93	&	1.70	&	38	&	68	&	931	&	11	\\
	&		&	76	&	78	&	0.21	&	1.05	&	34	&	120	&	565	&	10	\\
	&		&	88	&	90	&	0.25	&	1.06	&	47	&	65	&	857	&	11	\\
	&		&	61	&	61	&	0.15	&	0.62	&	28	&	59	&	900	&	11	\\
	&		&	66	&	66	&	0.11	&	0.37	&	44	&	52	&	706	&	10	\\
\cmidrule{3-10} 																			
	&		&\bf	77	&\bf	78	&\bf	0.24	&\bf	0.85	&\bf	37	&\bf	70	&\bf	779	&\bf	11	\\
\cmidrule{3-10} 																			
	&	200	&	44	&	45	&	0.10	&	0.27	&	22	&	38	&	926	&	11	\\
	&		&	32	&	33	&	0.06	&	0.27	&	15	&	38	&	571	&	10	\\
	&		&	49	&	50	&	0.13	&	0.29	&	26	&	27	&	860	&	11	\\
	&		&	38	&	38	&	0.17	&	0.32	&	18	&	38	&	903	&	11	\\
	&		&	44	&	45	&	0.07	&	0.30	&	29	&	31	&	717	&	11	\\
\cmidrule{3-10} 
	&		&\bf	41	&\bf	42	&\bf	0.10	&\bf	0.29	&\bf	21	&\bf	34	&\bf	783	&\bf	11	\\
\cmidrule{3-10} 																			
	&	300	&	57	&	57	&	0.49	&	0.75	&	21	&	35	&	936	&	11	\\
	&		&	33	&	34	&	0.05	&	0.32	&	15	&	37	&	568	&	10	\\
	&		&	48	&	49	&	0.05	&	0.19	&	24	&	24	&	862	&	11	\\
	&		&	37	&	38	&	0.08	&	0.30	&	16	&	35	&	906	&	11	\\
	&		&	43	&	43	&	0.08	&	0.22	&	27	&	28	&	724	&	11	\\
\cmidrule{3-10} 																			
	&		&\bf	43	&\bf	43	&\bf	0.10	&\bf	0.31	&\bf	20	&\bf	31	&\bf	786	&\bf	11	\\
\cmidrule{3-10} 																			
160	&	500	&	13584	&	13862	&	49.30	&	177.68	&	1502	&	3274	&	2272	&	20	\\
	&		&	15992	&	15790	&	14.80	&	69.07	&	2469	&	2807	&	2102	&	22	\\
	&		&	11721	&	11857	&	8.12	&	89.31	&	2030	&	1858	&	2176	&	19	\\
	&		&	5354	&	5442	&	8.91	&	90.58	&	819	&	1319	&	4600	&	14	\\
	&		&	2384	&	2384	&	2.89	&	17.53	&	364	&	572	&	2366	&	22	\\
\cmidrule{3-10} 																			
	&		&\bf	7987	&\bf	8043	&\bf	10.88	&\bf	70.49	&\bf	1176	&\bf	1667	&\bf	2575	&\bf	19	\\
\cmidrule{3-10} 																			
	&	2000	&	10933	&	11088	&	34.37	&	165.50	&	1219	&	2484	&	2272	&	14	\\
	&		&	12505	&	12489	&	13.33	&	59.33	&	1954	&	2156	&	2111	&	12	\\
	&		&	10623	&	10606	&	6.35	&	74.90	&	1847	&	1448	&	2193	&	14	\\
	&		&	4846	&	5005	&	8.66	&	84.26	&	709	&	958	&	4597	&	14	\\
	&		&	2028	&	2062	&	2.28	&	14.83	&	304	&	461	&	2379	&	13	\\
\cmidrule{3-10} 																			
	&		&\bf	6775	&\bf	6857	&\bf	8.95	&\bf	62.04	&\bf	989	&\bf	1280	&\bf	2583	&\bf	13	\\
\cmidrule{3-10} 																			
	&	3000	&	7775	&	7852	&	12.58	&	79.05	&	1193	&	2418	&	2274	&	13	\\
	&		&	11940	&	12051	&	13.67	&	62.07	&	1900	&	2090	&	2109	&	13	\\
	&		&	9872	&	10298	&	7.32	&	74.13	&	1817	&	1406	&	2195	&	13	\\
	&		&	4451	&	4532	&	6.56	&	68.13	&	696	&	917	&	4604	&	14	\\
	&		&	1994	&	2012	&	2.41	&	14.59	&	297	&	449	&	2382	&	13	\\
\cmidrule{3-10} 																			
	&\		&\bf	6054	&\bf	6162	&\bf	7.24	&\bf	51.48	&\bf	968	&\bf	1240	&\bf	2585	&\bf	13	\\
 \bottomrule
\end{tabular} }
\end{center}
\end{table}

\begin{table}
\begin{center}
\scriptsize
\caption{BB tests: algorithm of \citet{adelgren2016}.}
\label{table:AGold}
\begin{tabular}{rrr|rr|rr|r|rr|rr|r}
\toprule
  &   \multicolumn{2}{c|}{Time} & \multicolumn{2}{c|}{Time (insert)} & \multicolumn{2}{c|}{Time (other)} & \multicolumn{1}{c|}{\#BBnodes} & \multicolumn{2}{c|}{\#inserts} & \multicolumn{2}{c|}{\#nodes} & tree  \\ 
Size & T & L& T & L& T & L &   & T & L& T & L& depth \\ 
\midrule
\multicolumn{13}{c}{Instances from \citet{belotti2012biobjective} (averaged)}\\
\midrule
	60	&	5	&	5	&	0.00	&	0.00	&	0.00	&	0.01	&	76	&	189	&	   	&	65	&	64	&	6	\\
	80	&	13	&	13	&	0.00	&	0.00	&	0.01	&	0.01	&	109	&	247	&	   	&	83	&		&	7	\\
\midrule
\multicolumn{13}{c}{Instances from \citet{boland2013biobjective}, type 1}\\
\midrule
	80	&	10	&	11	&	0.01	&	0.09	&	0.01	&	0.15	&	363	&	6429	&	    	&	943	&	930	&	11	\\
		&	5	&	5	&	0.02	&	0.03	&	0.00	&	0.13	&	220	&	5066	&	    	&	590	&	583	&	12	\\
		&	18	&	19	&	0.01	&	0.16	&	0.02	&	0.45	&	580	&	9491	&	    	&	885	&	883	&	10	\\
		&	6	&	7	&	0.02	&	0.05	&	0.01	&	0.15	&	235	&	4919	&	    	&	928	&	920	&	11	\\
		&	7	&	7	&	0.02	&	0.07	&	0.01	&	0.09	&	269	&	4795	&	    	&	737	&	732	&	10	\\
\cmidrule{2-13}																											
		&\bf	8	&\bf	8	&\bf	0.02	&\bf	0.07	&\bf	0.01	&\bf	0.16	&\bf	311	&\bf	5923	&\bf	    	&\bf	804	&\bf	797	&\bf	11	\\
\cmidrule{2-13}																											
	160	&	205	&	211	&	0.12	&	1.47	&	0.09	&	3.44	&	1418	&	45501	&	     	&	2293	&	2285	&	17	\\
		&	129	&	139	&	0.09	&	1.91	&	0.12	&	2.80	&	997	&	41534	&	     	&	2160	&	2136	&	13	\\
		&	104	&	107	&	0.07	&	0.64	&	0.08	&	1.59	&	924	&	29895	&	     	&	2250	&	2203	&	13	\\
		&	400	&	414	&	0.30	&	4.86	&	0.13	&	9.76	&	2315	&	79421	&	     	&	4785	&	4675	&	17	\\
		&	108	&	111	&	0.03	&	0.82	&	0.07	&	1.94	&	743	&	22397	&	     	&	2489	&	2408	&	13	\\
\cmidrule{2-13}																											
		&\bf	164	&\bf	170	&\bf	0.09	&\bf	1.48	&\bf	0.10	&\bf	3.11	&\bf	1176	&\bf	39850	&\bf	     	&\bf	2658	&\bf	2610	&\bf	14	\\
\cmidrule{2-13}																											
	320	&	6266	&	6713	&	1.70	&	205.34	&	1.16	&	189.91	&	6436	&	486497	&	487183	&	11184	&	10970	&	48	\\
		&	4892	&	5457	&	1.84	&	243.26	&	0.74	&	242.09	&	5278	&	354556	&	354343	&	12070	&	11689	&	18	\\
		&	2742	&	2924	&	1.41	&	105.24	&	0.30	&	65.32	&	3367	&	246820	&	      	&	12488	&	12138	&	15	\\
		&	5291	&	5629	&	1.56	&	176.53	&	0.57	&	132.42	&	5925	&	364204	&	363980	&	12863	&	12531	&	16	\\
		&	2404	&	2522	&	1.38	&	49.21	&	0.43	&	60.80	&	3773	&	223789	&	      	&	9955	&	9703	&	18	\\
\cmidrule{2-13}																											
		&\bf	4035	&\bf	4329	&\bf	1.57	&\bf	135.49	&\bf	0.58	&\bf	119.31	&\bf	4803	&\bf	322156	&\bf	322168	&\bf	11664	&\bf	11361	&\bf	21	\\
\midrule
\multicolumn{13}{c}{Instances from \citet{boland2013biobjective}, type 2}\\
\midrule
	800	&	1	&	1	&	0.00	&	0.00	&	0.00	&	0.00	&	39	&	254	&	   	&	54	&	  	&	6	\\
		&	2	&	2	&	0.00	&	0.00	&	0.01	&	0.00	&	46	&	279	&	   	&	64	&	  	&	7	\\
		&	4	&	4	&	0.00	&	0.00	&	0.00	&	0.00	&	98	&	589	&	   	&	90	&	  	&	7	\\
		&	9	&	9	&	0.00	&	0.01	&	0.01	&	0.01	&	180	&	900	&	   	&	129	&	128	&	7	\\
\cmidrule{2-13}																											
		&\bf	3	&\bf	3	&\bf	0.00	&\bf	0.00	&\bf	0.00	&\bf	0.00	&\bf	75	&\bf	440	&\bf	   	&\bf	80	&\bf	79	&\bf	7	\\
\cmidrule{2-13}																											
	1250	&	9	&	9	&	0.00	&	0.00	&	0.01	&	0.01	&	147	&	867	&	   	&	140	&	   	&	8	\\
		&	19	&	19	&	0.00	&	0.00	&	0.00	&	0.02	&	310	&	1922	&	1935	&	200	&	201	&	8	\\
		&	19	&	19	&	0.02	&	0.00	&	0.02	&	0.06	&	302	&	1962	&	    	&	245	&	   	&	9	\\
		&	46	&	46	&	0.01	&	0.01	&	0.00	&	0.06	&	542	&	3328	&	    	&	281	&	   	&	9	\\
\cmidrule{2-13}																											
		&\bf	20	&\bf	20	&\bf	0.01	&\bf	0.00	&\bf	0.01	&\bf	0.03	&\bf	294	&\bf	1816	&\bf	1819	&\bf	210	&\bf	210	&\bf	8	\\
\cmidrule{2-13}																											
	2500	&	514	&	515	&	0.04	&	0.13	&	0.29	&	1.42	&	2893	&	7650	&	    	&	306	&	   	&	9	\\
		&	944	&	945	&	0.01	&	0.18	&	0.31	&	2.28	&	3831	&	9409	&	9406	&	370	&	   	&	9	\\
		&	1243	&	1244	&	0.02	&	0.21	&	0.45	&	2.73	&	4007	&	9849	&	    	&	451	&	449	&	10	\\
		&	3185	&	3199	&	0.02	&	0.45	&	0.86	&	6.24	&	8468	&	18665	&	     	&	386	&	385	&	9	\\
\cmidrule{2-13}																											
		&\bf	1177	&\bf	1180	&\bf	0.02	&\bf	0.22	&\bf	0.43	&\bf	2.73	&\bf	4404	&\bf	10725	&\bf	10724	&\bf	375	&\bf	374	&{\bf	9}	\\
 \bottomrule
\end{tabular}
\end{center}
\end{table}

Tables \ref{table:BSWold} and \ref{table:AGold} show that the time spent by the BoT is much less than that spent using the linked list. The total BB time shadows this dominance as insertion time is a negligible fraction of the total solve time. However, these results confirm those of experiment 1: even within a BB, insertion into a linked list often takes orders of magnitude longer than for a BoT. For all instances of the two considered classes, the number of insertions amounts to up to a few million yet the final structure only has a few thousand solutions (see e.g.\mbox{} Table \ref{table:BSWold}, instances by \citet{boland2013biobjective}, for size 160); this suggests a pattern similar to that with larger values of $\mu$ (see Tables \ref{table_results_a}-\ref{table_results_c}). 

Both solvers that we used are rather rudimentary implementations---the alternative method by \citet{boland2013biobjective} easily outperforms the BB solver by \citet{belotti2012biobjective}. This explains why the BB times reported can be hours even if the BB algorithms visit up to only a few thousands of BB nodes, well below the current state-of-the-art BB algorithms, where millions, or tens of millions, of nodes can be explored in the same time for much larger instances.  

\begin{table}
\begin{center}
\scriptsize
\caption{
BB tests: algorithm of \citet{belotti2012biobjective} on geometrical instances. 
} \label{table:BSWgraph}
\begin{tabular}{rrrr|rr|r|r|r|rrr}
\toprule
&  & \multicolumn{2}{c}{Time} & \multicolumn{2}{c}{Time(insert)} & \multicolumn{1}{c}{\#BBnodes}  & \multicolumn{1}{c}{\#inserts}          & \multicolumn{1}{c}{\#nodes}     & tree\\ 
$N$ & $k$ & T & L & T & L & \multicolumn{1}{c}{($\times 10^3$)} & \multicolumn{1}{c}{($\times 10^3$)} & \multicolumn{1}{c}{($\times 10^3$)} & depth\\
\midrule
500	&	1	&	15	&	25	&	0.8	&	6.8	&	10	&	35	&	5	&	24	\\
	&	5	&	15	&	16	&	0.6	&	2.3	&	10	&	35	&	5	&	23	\\
	&	25	&	21	&	35	&	1.0	&	6.2	&	10	&	34	&	5	&	25	\\
	&	125	&	48	&	47	&	0.5	&	2.2	&	10	&	33	&	5	&	27	\\
\cmidrule{3-10} 																			
	&		&\bf	22	&\bf	28	&\bf	0.7	&\bf	3.8	&\bf	10	&\bf	34	&\bf	5	&\bf	25	\\
\cmidrule{3-10} 																			
1000	&	1	&	36	&	39	&	2.5	&	10.3	&	20	&	70	&	10	&	25	\\
	&	5	&	39	&	41	&	2.5	&	11.2	&	20	&	70	&	10	&	26	\\
	&	25	&	47	&	58	&	2.6	&	13.8	&	20	&	69	&	10	&	28	\\
	&	125	&	103	&	115	&	3.1	&	16.6	&	20	&	67	&	10	&	37	\\
\cmidrule{3-10} 																			
	&		&\bf	51	&\bf	57	&\bf	2.6	&\bf	12.7	&\bf	20	&\bf	69	&\bf	10	&\bf	29	\\
\cmidrule{3-10} 																			
2000	&	1	&	111	&	185	&	9.0	&	95.7	&	40	&	140	&	20	&	23	\\
	&	5	&	100	&	186	&	8.3	&	96.2	&	40	&	140	&	20	&	30	\\
	&	25	&	133	&	212	&	9.2	&	96.9	&	40	&	139	&	20	&	32	\\
	&	125	&	224	&	345	&	9.7	&	113.7	&	40	&	137	&	20	&	40	\\
\cmidrule{3-10} 																			
	&		&\bf	135	&\bf	224	&\bf	9.0	&\bf	100.4	&\bf	40	&\bf	139	&\bf	20	&\bf	31	\\
\cmidrule{3-10} 																			
4000	&	1	&	370	&	966	&	41.9	&	681.9	&	80	&	280	&	40	&	32	\\
	&	5	&	350	&	1070	&	53.0	&	774.3	&	80	&	280	&	40	&	31	\\
	&	25	&	409	&	1025	&	42.6	&	691.1	&	80	&	279	&	40	&	30	\\
	&	125	&	650	&	1308	&	48.4	&	751.8	&	80	&	277	&	40	&	27	\\
\cmidrule{3-10} 																			
	&		&\bf	431	&\bf	1085	&\bf	46.2	&\bf	723.7	&\bf	80	&\bf	279	&\bf	40	&\bf	30	\\
\cmidrule{3-10} 																			
8000	&	1	&	1952	&	6275	&	339.2	&	5185.5	&	160	&	560	&	80	&	39	\\
	&	5	&	1861	&	4786	&	289.5	&	3893.4	&	160	&	560	&	80	&	38	\\
	&	25	&	2113	&	5363	&	329.6	&	4339.8	&	160	&	559	&	80	&	34	\\
	&	125	&	3793	&	6189	&	329.2	&	4679.2	&	160	&	556	&	80	&	89	\\
\cmidrule{3-10} 																			
	&		&\bf	2323	&\bf	5619	&\bf	321.3	&\bf	4499.8	&\bf	160	&\bf	559	&\bf	80	&\bf	46	\\
\cmidrule{3-10} 																			
16000	&	1	&	8973	&	22803	&	1627.5	&	19345.8	&	320	&	1120	&	160	&	41	\\
	&	5	&	8059	&	26171	&	1638.7	&	22958.1	&	320	&	1120	&	160	&	41	\\
	&	25	&	12250	&	23252	&	1970.6	&	20019.5	&	320	&	1119	&	160	&	41	\\
	&	125	&	13464	&	29552	&	1674.1	&	24950.9	&	320	&	1116	&	160	&	118	\\
\cmidrule{3-10} 																			
	&		&\bf	10450	&\bf	25305	&\bf	1722.3	&\bf	21702.8	&\bf	320	&\bf	1119	&\bf	160	&\bf	53	\\
 \bottomrule
\end{tabular} 
\end{center}
\end{table}

\begin{table}
\begin{center}
\scriptsize
\caption{
BB tests: algorithm of \citet{adelgren2016} on geometrical instances. 
} \label{table:AGgraph}
\begin{tabular}{rrrr|rr|rr|r|r|r|r}
\toprule
&  & \multicolumn{2}{c}{Time} & \multicolumn{2}{c}{Time (insert)} & \multicolumn{2}{c}{Time (other)} & \multicolumn{1}{c}{\#BBnodes}       & \multicolumn{1}{c}{\#inserts} & \multicolumn{1}{c}{\#nodes} & tree\\ 
$N$ & $k$& T & L& T & L	& T & L	& \multicolumn{1}{c}{($\times 10^3$)} & \multicolumn{1}{c}{($\times 10^3$)}& \multicolumn{1}{c}{($\times 10^3$)} & depth \\
\midrule
	500	&	1	&	12	&	56	&	1.93	&	16.13	&	0.09	&	23.85	&	10	&	20	&	5	&	16	\\
		&	5	&	119	&	158	&	1.08	&	11.37	&	0.16	&	21.86	&	10	&	20	&	5	&	15	\\
		&	25	&	131	&	160	&	0.85	&	9.59	&	0.12	&	18.19	&	10	&	20	&	5	&	15	\\
		&	125	&	163	&	200	&	0.93	&	13.19	&	0.14	&	20.78	&	10	&	20	&	5	&	14	\\
\cmidrule{3-12}																								
		&		&\bf	75	&\bf	130	&\bf	1.13	&\bf	12.34	&\bf	0.12	&\bf	21.07	&\bf	10	&\bf	20	&\bf	5	&\bf	15	\\
\cmidrule{3-12}																								
	1000	&	1	&	39	&	273	&	5.31	&	90.61	&	0.16	&	104.85	&	20	&	40	&	10	&	16	\\
		&	5	&	458	&	707	&	3.13	&	60.20	&	0.45	&	133.60	&	20	&	40	&	10	&	15	\\
		&	25	&	511	&	762	&	3.09	&	66.47	&	0.38	&	135.67	&	20	&	40	&	10	&	15	\\
		&	125	&	619	&	895	&	3.08	&	69.18	&	0.38	&	149.48	&	20	&	40	&	10	&	16	\\
\cmidrule{3-12}																								
		&		&\bf	273	&\bf	603	&\bf	3.55	&\bf	70.77	&\bf	0.32	&\bf	129.83	&\bf	20	&\bf	40	&\bf	10	&\bf	15	\\
\cmidrule{3-12}																								
	2000	&	1	&	133	&	1346	&	14.91	&	389.46	&	0.50	&	546.32	&	40	&	80	&	20	&	17	\\
		&	5	&	1837	&	3314	&	16.28	&	349.33	&	1.45	&	756.48	&	40	&	80	&	20	&	18	\\
		&	25	&	2032	&	3251	&	13.38	&	305.68	&	1.54	&	607.06	&	40	&	81	&	21	&	20	\\
		&	125	&	2298	&	3725	&	13.47	&	316.52	&	1.68	&	760.57	&	40	&	81	&	20	&	18	\\
\cmidrule{3-12}																								
		&		&\bf	1034	&\bf	2711	&\bf	14.46	&\bf	338.72	&\bf	1.17	&\bf	660.93	&\bf	40	&\bf	81	&\bf	20	&\bf	18	\\
\cmidrule{3-12}																								
	4000	&	1	&	996	&	6695	&	80.32	&	1596.64	&	3.07	&	2645.33	&	80	&	160	&	40	&	18	\\
		&	5	&	8530	&	15352	&	106.16	&	1518.61	&	8.59	&	3974.04	&	80	&	160	&	40	&	18	\\
		&	25	&	8273	&	14081	&	99.59	&	1423.79	&	5.36	&	3189.86	&	80	&	162	&	42	&	18	\\
		&	125	&	10466	&	16394	&	91.13	&	1386.03	&	7.60	&	3239.40	&	80	&	160	&	40	&	18	\\
\cmidrule{3-12}																								
		&		&\bf	5208	&\bf	12411	&\bf	93.79	&\bf	1479.00	&\bf	5.73	&\bf	3228.40	&\bf	80	&\bf	160	&\bf	40	&\bf	18	\\
 \bottomrule
\end{tabular} 
\end{center}
\end{table}

The advantages of BoT over list are more apparent on geometrical instances, which, albeit simple, allow for a thorough stress test of both the BB solvers and the data structures. For these instances, the time spent by the data structures is a significant portion of the total CPU time. All BB nodes are solved very quickly, and as a result many solutions are found and many BB nodes are explored in a short time. Inserting all of these solutions, while efficient for a BoT, requires much more computational effort for a linked list, to the point that the amount of time spent inserting solutions in a large linked list becomes a substantial percentage of the total CPU time.

These results suggest that an efficient data structure will be essential when faster and more stable BB solvers, for instance using quicker fathoming rules \citep{belotti2016fathoming}, become available: as the time for each BB node decreases and the number of inserted solutions increases, the data structure must be efficient enough that the insertion function only takes a small percentage of the total CPU time.


\section{Concluding remarks}

We have introduced the {\em biobjective tree}, a variant of a binary tree to efficiently store sets of nondominated solutions of BOMILPs. The BoT is equipped with an insertion procedure for adding  points and possibly eliminating several other points that are dominated by the newly inserted one. The BoT also has the desirable property that an {\em in-order} pass produces a sorted list of nondominated points. We tested the practical value of the BoT with two experiments. The results show that a BoT provides a more efficient method for storing solutions to BOMILP than a list of points. They also show that a BoT is also a very useful tool when used in BB methods for solving BOMILPs.

It is worth noting here that we have focused on a simple BST data structure, instead of one with better balancing properties, because we wanted to avoid the travails of studying and implementing complicated rebalancing procedures, mixed with the already taxing duty of maintaining other desirable properties of a BoT, i.e., the aforementioned ability to hold nondominated points at any step.

Data structures like the BoT are preferable over lists when inserting large numbers of points in an undefined order, such as in a BB algorithm. Because current state-of-the-art BB algorithms for BOMILP can tackle relatively small problems, one can observe the impact of the data structure in those classes of problems where BB node solution is fast, i.e., those problems that admit a simple and compact structure yet contain a very large number of nondominated points. We speculate that future implementations of a BB algorithm for BOMILP will expose even more clearly the advantages of a BST data structure. However, algorithms that find solutions in a non-random order might benefit from more specific, and perhaps less sophisticated, data structures.


\bibliographystyle{plainnat}
\bibliography{datastruct}

\end{document}